\documentclass[twoside]{article}

\usepackage[accepted]{aistats2026}
\usepackage[round]{natbib}

\usepackage[utf8]{inputenc} 
\usepackage[T1]{fontenc}    
\usepackage{hyperref}       
\usepackage{url}            
\usepackage{booktabs}       
\usepackage{amsfonts}       
\usepackage{nicefrac}       
\usepackage{microtype}      
\usepackage{xcolor}         
\usepackage{float}

\usepackage{amsmath}
\usepackage{amssymb}
\usepackage{mathtools}
\usepackage{amsthm}

\usepackage{dblfloatfix}
\usepackage{caption}
\usepackage{tabularray}
\usepackage{adjustbox}
\usepackage{graphicx}
\usepackage{dsfont}
\usepackage{svg}

\usepackage{algorithm}
\usepackage{algpseudocode}
\usepackage{algorithmicx}

\usepackage{xfrac}

\theoremstyle{plain}
\newtheorem{theorem}{Theorem}[section]

\newtheorem{lemma}[theorem]{Lemma}
\newtheorem{corollary}[theorem]{Corollary}
\theoremstyle{definition}
\newtheorem{definition}[theorem]{Definition}

\theoremstyle{remark}
\newtheorem{remark}[theorem]{Remark}

\newcommand{\cN}{\mathcal{N}}

\newcommand{\RR}{\mathbb{R}}
\newcommand{\EE}{\mathbb{E}}

\newcommand{\PP}{\mathbb{P}}

\newcommand{\F}{{\textup{F}}}

\newcommand{\TLSE}{\texttt{T-LSE}}
\newcommand{\RLSE}{\texttt{R-LSE}}
\newcommand{\RSC}{\texttt{RSC}}
\newcommand{\LSE}{\texttt{LSE}}

\newcommand{\reusefootnote}[1]{%
  \textsuperscript{\hyperref[#1]{\ref*{#1}}}%
}

\hypersetup{
  colorlinks = true,
  urlcolor = blue,
  linkcolor = blue,
  citecolor = blue
}

\begin{document}

%
\runningtitle{Rank Adaptive Inference}

%

\twocolumn[

\aistatstitle{Near-optimal Rank Adaptive Inference of \\ High Dimensional Matrices}

\aistatsauthor{ Frédéric Zheng  \And Yassir Jedra \And  Alexandre Proutière }

\aistatsaddress{ KTH\\
	Stockholm, Sweden \And  MIT\\
	Cambridge, MA, USA \And KTH\\
	Stockholm, Sweden } ]

\begin{abstract}
  We address the problem of estimating a high-dimensional matrix from linear measurements, with a focus on designing optimal rank-adaptive algorithms. These algorithms infer the matrix by estimating its singular values and the corresponding singular vectors up to an effective rank, adaptively determined based on the data. We establish, for the first time, instance-specific lower bounds for the sample complexity of such algorithms. We uncover fundamental trade-offs in selecting the effective rank: balancing the precision of estimating a subset of singular values against the approximation cost incurred for the remaining ones. Our analysis identifies how the optimal effective rank depends on the matrix being estimated, the sample size, and the noise level. We propose an algorithm that combines a Least-Squares estimator with a universal singular value thresholding procedure. We provide finite-sample error bounds for this algorithm, that are tighter than those of existing rank-adaptive algorithms. Furthermore, our bounds nearly match the derived fundamental limits. Finally, we confirm experimentally that our algorithm outperforms existing rank-adaptive algorithms.
\end{abstract}

\section{INTRODUCTION}

We revisit the {canonical} problem of estimating a high-dimensional matrix from linear measurements. The learner has access to $n$ samples, {$(x_1,y_1), \dots, (x_n, y_n)$ where $y_i\in \mathbb{R}^{d_y}$ is a noisy realization of $Ax_i$ with  $x_i\in \mathbb{R}^{d_x}$. $A \in \mathbb{R}^{d_y \times d_x}$  is considered a priori unknown}. The objective is to estimate the matrix $A$ as accurately as possible, and more precisely to construct an estimator $\hat{A}_n$ with minimal Frobenius error $\Vert \hat{A}_n - A\Vert_{\F}$  with high probability. This problem has widespread applications across various domains, including healthcare, biology, computer vision, and control (see \cite{wu2020adaptive} for a non-exhaustive list of examples). In this paper, we illustrate our results through two key applications: (1) multivariate regression where the covariates {$(x_i)_{i \in \{1,\ldots,n\}}$} form an i.i.d. sequence of random variables and (2) linear system identification where the covariates are the successive states of a linear time-invariant dynamical system governed by $A$, meaning that $x_{i+1}$ is a noisy version of $Ax_i$.


The high-dimensional setting arises when the  number of entries of $A$ are comparable to or exceed the number of available observations. In such cases, imposing and leveraging additional structural properties — such as sparsity or low rank — is essential for statistically sound matrix estimation. Ideally, an estimation algorithm should automatically detect whether such a structure exists and construct an estimator accordingly. In this work, we focus on scenarios where the relevant structure is the matrix rank and aim to design optimal rank-adaptive algorithms. These algorithms estimate the matrix by inferring its singular values and corresponding singular vectors up to an effective rank, which is adaptively determined based on the data. The key questions we explore are: What are the fundamental performance limits of such algorithms? Can we devise an algorithm approaching these limits? How does the optimal effective rank depend on the matrix being estimated, the covariates, and the sample size? We answer these questions with, to the best of our knowledge, an unprecedented level of precision. Specifically, our contributions are as follows.

\emph{1. Instance-specific Sample Complexity Lower bounds.} We establish the first instance-specific lower bounds for the sample complexity\footnote{By sample complexity here, we mean the minimum number of samples required for the existence of an $(\varepsilon,\delta)$-PAC rank-adaptive algorithm (i.e., returning an estimator $\hat{A}_n$ such that $\| \hat{A}_n - A\|_F\le \varepsilon$ with probability at least $1-\delta$.)} of rank-adaptive algorithms. These bounds reveal fundamental trade-offs in selecting the effective rank and precisely depend on both the matrix $A$ and the covariates' {distribution}. 
For instance, in the case of multivariate regression, our results show that with $n$ samples and with $d_x=d_y$, the Frobenius norm error scales at least as:

$$
\min_{k} \left( \sigma^2\frac{\log (\frac{1}{\delta})+kd_x}{n\underline{\lambda}_k(\Sigma)} +\sum_{i>k}s_i^2(A) \right),
$$
where $\sigma$ is the noise level, $s_i(A)$ denotes the $i$-th largest singular value of $A$, $\Sigma$ is the covariance matrix of covariates, and  $\underline{\lambda}_k(\Sigma)$ is the average of its $k$ smallest eigenvalues. Furthermore, if we were able to learn the value of $k$ that minimizes this bound, we would directly obtain the optimal effective rank. The proof techniques towards our lower bounds are of independent interest. They intricately combine change-of-measure arguments, which provide explicit dependence on $A$ and the covariates, with packing arguments, which capture the dependence on the dimensions $d_x$ and $d_y$.

\emph{2. Thresholded Least-squares Estimation is nearly optimal.} We propose an algorithm that combines a least squares estimator with a universal singular value thresholding procedure. We derive finite-sample error bounds for this estimator and show that it outperforms existing algorithms, while closely approaching the fundamental performance limits. Our analysis builds on an improved understanding of matrix denoising techniques based on singular value thresholding. Specifically, we demonstrate that applying universal singular value thresholding to an existing estimator generally enhances its performance guarantees. We illustrate this principle using the estimator obtained via nuclear norm regularization and the Least-Squares estimator.

\emph{3. Applications and Numerical Results.} We exemplify our theoretical results to both multivariate regression and linear system identification. Finally, we present numerical experiments to complement and validate our theoretical findings.

\paragraph{Notation.} We use $a \lesssim b$ (resp. $a \gtrsim b$) to mean that $a$ is smaller (resp. larger) than $b$ up to a universal multiplicative constant.
 We use $a \wedge b$ (resp. $a \vee b$) to denote $\min(a, b)$ (resp. $\max(a,b)$). Let $[d]=\{1,...,d\}$. Given a matrix $M \in \RR^{d_x \times d_y}$, $\Vert M \Vert_\F$ denotes its Frobenius norm, $\Vert M \Vert_2$ denotes its operator norm, and $\Vert M \Vert_1$ denotes its nuclear norm. Let $\underline{d}=d_x\wedge d_y$ and $\bar d=d_x\vee d_y$. The singular values (resp. eigenvalues) of $M$ are denoted by $s_1(M), \dots, s_{\underline{d}}(M)$ (resp. $\lambda_{1}(M), \dots, \lambda_{\underline{d}}(M)$) in decreasing order, and its condition number is denoted by $\kappa(M) = s_1(M)/s_d(M)$. For $k\le \underline{d}$, we use $\Pi_k(M)$ to denote the best rank-$k$ approximation of $M$. $M^\dagger$ is the Moore-Penrose pseudo-inverse of $M$. When $M$ is symmetric positive semi-definite, we denote its maximum (resp. minimum) eigenvalue by $\lambda_{\max}(M)$ (resp. $\lambda_{\min}(M)$). We also define $\bar{\lambda}_k(M) := \sum_{i=1}^k \lambda_i(M)/k$ and $\underline{\lambda}_k(M) := \sum_{i=0}^{k-1} \lambda_{d - i}(M)/k$.


\section{PRELIMINARIES}\label{sec:prelim}

{\bf Model and objective.} We are given $n$ observations, $(x_1, y_1), (x_2, y_2), \dots, (x_n, y_n)$, taking values in $\RR^{d_x \times d_y}$. We consider a linear model whereby $y_i = Ax_i +\eta_i$ for all $i \in[n]=\{1,\ldots,n\}$. $A $ is a matrix in $\RR^{d_y \times d_x}$ that is a priori unknown. $(\eta_i)_{i \ge 1}$ is a sequence of i.i.d. zero-mean, $\sigma^2$-sub-Gaussian random variables taking values in $\RR^{d_y}$ where $\sigma> 0$ denotes their variance proxy parameter. Our objective is to estimate the matrix $A$.

We define the empirical covariance matrix as $\hat{\Sigma} = (\sum_{i=1}^n x_i x_i^\top )/n$. We assume that $\EE[\hat{\Sigma}] \succ 0${, ensuring identifiability of $A$ by least squares.}
We are interested in two specific settings where this assumption naturally holds:

{\it 1. Multivariate regression.} Here the covariates $x_1, \dots, x_n$ are assumed to be i.i.d. random vectors with covariance matrix $\Sigma := \EE[\hat{\Sigma}] = \EE[x_1 x_1^\top]$.


{\it 2. Linear system identification.} Here, $y_i = x_{i+1}$ for all $i \ge 1$, and thus the covariates $x_1, \dots, x_n$ evolve according to the dynamical system: $x_{t+1} = A x_{t} + \eta_t$ for all $t\ge 1$, with $x_1 = 0$. In this case, $d_x = d_y$. For all $t \ge 1$, the finite-time controllability Gramian of the system as $\Gamma_p(A):=\sum_{k=0}^p (A^k) (A^k)^\top$. We assume that the system is stable (i.e., $\rho(A) < 1$, where $\rho(A)$ denotes the spectral radius of $A$). Thus,  we may define $\Gamma_\infty(A)$ as $ \lim_{t \to \infty} \Gamma_t(A)$.


In what follows, we will often prefer to use matrix notations and write $Y = X A^\top + E$ where  $Y^\top = \begin{bmatrix}
    y_1 & \!\! \cdots & \!\!y_n
\end{bmatrix} \in \RR^{d_y \times n}$, $X^\top  = \begin{bmatrix}
    x_1 & \!\!\cdots & \!\!x_n
\end{bmatrix} \in \RR^{d_x \times n}$ and $E^\top  = \begin{bmatrix}
    \eta_1 & \!\!\cdots & \!\!\eta_n
\end{bmatrix} \in \RR^{d_y \times n}$.

\section{RELATED WORK}


The inference of high-dimensional matrices using reduced-rank regression has a long history and has attracted significant attention over the past decades (see, e.g., \cite{anderson1951estimating,IZENMAN1975,Reinsel1998,anderson1999asymptotic,negahban2011estimation,bunea2011optimal}). It also has numerous connections to iconic problems such as principal component analysis and matrix completion \cite{candes2009, koltchinskii2011, chatterjee2015matrix}. Below, we present a selection of papers that we consider most relevant to our analysis. We begin by discussing existing fundamental limits before moving on to algorithms with finite sample-size performance guarantees.


{\bf Fundamental limits.} For multivariate regression and matrix completion, most existing lower bounds on the estimation error are minimax and concern matrices with known rank (or with a known upper bound on its rank). For instance, the lower bounds derived in \cite{Rohde2009EstimationOH} are minimax and assume that the design matrix $X$ satisfies some properties such as the Restricted Isometry Property (RIP), and that the rank is fixed. \cite{candes2011tight} derived lower bounds satisfied when the design matrix $X$ has the RIP property, and that remain valid for approximately low-rank matrices. However, these bounds are in expectation, and do not cover fully adaptive algorithms. The lower bounds presented in \cite{bunea2011optimal} are based on those from \cite{Rohde2009EstimationOH, candes2011tight} and hence suffer from the same shortcomings. \cite{cai2015optimal} considers the problem of optimal rank estimation for covariate matrices, and presents minimax rate detection limits, but these are only valid for exactly low-rank matrices. Closer to our analysis, \cite{wu2020adaptive} proposes an instance-specific (or more precisely locally minimax) lower bound, but unfortunately these are only valid for the risk and not the matrix estimation error. Finally, it is worth noting recent efforts towards the derivation of error lower bounds for the problem of linear system identification \cite{simchowitz2018learning,jedra2019,jedra2022finite, djehiche2021non, sun2023finite, bakshi2023new, zhang2024learning}. These bounds concern the estimation error in operator norm, they are not rank-adaptive, and most of them are minimax.

To the best of our knowledge, we are the first to derive precise instance-specific lower bounds that depend on the spectral properties (i.e., the singular values) of both the target matrix and the design matrix. This contrasts with traditional minimax lower bounds, which characterize performance limits for the worst-case matrix within a given class. As a result, minimax lower bounds fail to capture how algorithms can truly adapt to the specific matrix being estimated.


{\bf Algorithms and their finite-sample guarantees.} Most existing algorithms for reduced-rank regression rely on optimization methods that incorporate a penalty in the objective function to encourage low-rank solutions. Common choices for this penalty include the nuclear norm or its weighted variant \cite{yuan2007,candes2010matrix,negahban2011estimation,chen2013}, the rank itself \cite{bunea2011optimal}, and the Schatten-$p$ quasi-norm \cite{Rohde2009EstimationOH}. When the rank of $A$ is at most $r$, the most effective among these algorithms is the Rank-Constrained Selection (\RSC) algorithm introduced by \cite{bunea2011optimal}, and whose error bounds are analyzed in their Corollary 8: given $n$ samples, with probability at least $1-\delta$,
\begin{equation}\label{eq:soa1}
\Vert \hat{A}_n-A\Vert_\F^2 \lesssim \min_{k \in [r]} \Big(k\sigma^2\frac{ \log(\frac{1}{\delta})+\bar d }{n\lambda_\mathrm{min}(\hat \Sigma)} +
\kappa(\hat \Sigma)\sum_{i>k} s^2_i(A)\Big). 
\end{equation}

{The RSC algorithm is meant for minimizing the prediction error $\Vert X\hat{A}_n-XA\Vert_\F$. This explains the extra dependence on the condition number $\kappa(\hat \Sigma)$ in the bound on the identification error $\Vert \hat A_n - A\Vert_\F$.} Our algorithm, the Thresholded Least Squares Estimator (\TLSE), achieves provably better error bounds. It combines the \LSE ~with a universal singular value thresholding procedure, inspired by the seminal works of \cite{chatterjee2015matrix, gavish2014optimal}. {Finally, the design and analysis of algorithms with finite sample-size performance guarantees for identifying dynamical systems with low-rank structure, notably partially observed linear dynamical system, has also seen a surge of interest recently \cite{fazel2013hankel,djehiche2022efficient, sun2022finite, oymak2019non,simchowitz2019learning, sarkar2021finite, bakshi2023new}. However, to the best of our knowledge, we are the first to present rank-adaptive algorithms with guarantees, exploiting the low-rank structure of the state matrix.}

\section{INSTANCE-SPECIFIC SAMPLE COMPLEXITY LOWER BOUNDS}\label{sec:lb}

Our goal in this section is to derive fundamental limits on the number of samples required to obtain an $(\varepsilon,\delta)$-PAC estimate of the matrix $A$. To this end, we must focus on estimators or algorithms that genuinely adapt to the matrix $A$. For instance, an algorithm that always outputs $A$, regardless of the input data, would require no samples to be $(\varepsilon,\delta)$-PAC when the true matrix is $A$, but would fail for any matrix other than $A$. Therefore, to derive instance-specific sample complexity lower bounds, we must consider algorithms that are $(\varepsilon,\delta)$-PAC not just for $A$, but for all matrices in a {\it neighborhood} of $A$. The choice of this neighborhood involves a delicate trade-off. On the one hand, the neighborhood around $A$ should be small enough to ensure that the class of algorithms considered is broad enough and that the derived lower bound reflects the difficulty of estimating the specific instance $A$. On the other hand, if the neighborhood is too narrow, the resulting lower bound may be overly restrictive and potentially unattainable.



Deriving lower bounds on the sample complexity of {\it rank-adaptive} algorithms presents significant challenges. We address this by decomposing the problem into two steps:
(i) we first derive lower bounds (see \textsection\ref{subsec:lb_rc}) using a neighborhood of $A$ 
obtained by adding perturbations of specific rank $r\ge \mathrm{rank}(A)$. 
As shown later, this lower bound will be tight for {\it rank-constrained} algorithms, i.e., those returning matrices of rank $r$. 
(ii) Building on this analysis, we then derive refined lower bounds (see \textsection\ref{subsec:lb_adapt}) by considering a neighborhood of $A$ obtained by adding perturbations of rank bounded by the effective rank that optimally balances estimation and approximation errors. In Section \ref{sec:lse}, we introduce a rank-adaptive algorithm whose sample complexity {nearly} matches these lower bounds (without any prior knowledge), hence proving that our lower bounds are tight. 

Throughout this section, we assume that $(\eta_i)_{i \in [n]}$ are distributed according to $\cN(0, \sigma^2 I_{d_y})$. Proofs of the results presented in this section are in Appendix \ref{app:sec4}.

\subsection{Lower bounds for rank-constrained algorithms}\label{subsec:lb_rc}

As a starting point, we derive fundamental limits for algorithms estimating a matrix $A$ by a rank-$r$ matrix. We begin by presenting the following packing bounds for Stiefel manifolds, which will play an important role in the definition of the neighborhood of $A$. For any $k\le d$, the Stiefel manifold $\mathrm{St}_k^{d}(\mathbb{R})$ is defined as the set of semi-orthogonal matrices $\{Q\in \mathbb{R}^{d\times k}: \ Q^\top Q=I_k\}$. 

\begin{lemma}\label{lem:stiefel}
For any $k\leq d / 2$, there exists a finite $\sqrt{k}$-packing ${\cal P}_k^d$ of $\mathrm{St}_k^{d}(\mathbb{R})$ whose cardinality is larger than $2^{kd}$. By $\sqrt{k}$-packing, we mean that for any $Q, R \in  {\cal P}_{k}^d$, $Q\neq R$ implies $\Vert Q-R\Vert_\F\ge \sqrt{k}/C$ where $C\geq 1$ is a universal constant independent of $k$ and $d$. 
\end{lemma}

The explicit value of $C$ is given in Appendix \ref{app:sec4}. We define the neighborhood of $A$ as: 
${\cal C}(A,r,\varepsilon)={\cal C}_{1}(A,r,\varepsilon)\cup {\cal C}_{2}(A,r,\varepsilon)\cup \{A\} 
$ where
\begin{equation}
\left\{
\begin{array}{rl}
{\cal C}_{1}(A,r,\varepsilon) & =  \left \lbrace A+\frac{2C\varepsilon}{\sqrt{r}}QW_{-r}^\top:  Q\in {\cal P}_{r}^{d_y} \right \rbrace \\
{\cal C}_{2}(A,r,\varepsilon) & = \left\lbrace A+\frac{2C\varepsilon}{\sqrt{r}}U_rR^\top:  R\in {\cal P}_{r}^{d_x} \right\rbrace
\end{array}
\right..    
\end{equation}

Here, 
{$U_r\in \mathrm{St}_r^{d_y}(\mathbb{R})$ is a semi-orthogonal matrix whose columns $u_1,...,u_r$ are the left-singular vectors of $A$ corresponding to $s_1(A),...,s_r(A)$} while
$W_{-r}\in \mathrm{St}_r^{d_x}(\mathbb{R})$ contains the eigenvectors corresponding to the $r$ smallest eigenvalues of 1) $\Sigma$ for multivariate regression, 2) $\Gamma_\infty(A)$ for dynamical systems.  Our packing directly uses knowledge of the expected covariance matrix, avoiding the need for RIP assumption. We show below that this leads to a more precise dependency of the sample complexity with respect to the spectrum of $\Sigma$.

\begin{definition}[$(\varepsilon,\delta, r)$-stability]{\it Let $A$ such that $\mathrm{rank}(A)\le r$. 
An algorithm is $(\varepsilon,\delta, r)$-stable in $A$ if it returns a rank-$r$ matrix and the following assertion holds:
\begin{align}\label{eq:rc_stability}
    \exists N\in \mathbb{N}, \ \forall n\geq N, \ \forall A'\in {\cal C}(A,r,\varepsilon), \nonumber  \\\mathbb{P}_{ A'}(\Vert \hat A_n-A'\Vert _\F\leq\varepsilon )\geq 1-\delta.
\end{align} 
}
\vspace{-0.3cm}
\label{def:full}
\end{definition}
If an algorithm is $(\varepsilon,\delta, r)$-stable in $A$, we define its {\it sample complexity} as the minimal integer $N$ such that Assertion (\ref{eq:rc_stability}) holds. We tried with the above definition to construct a neighborhood $\mathcal{C}(A,r,\varepsilon)$ as small as possible. In fact, it is finite, which implies that the class of $(\varepsilon,r,\delta)$-stable algorithms in $A$ is very broad. In particular, it includes all algorithms returning rank-$r$ matrices and that are {\it consistent} in the sense that for any $A$, there exists a number of sample $N_A$ such that $\mathbb{P}_{ A}(\Vert \hat A_{N_A}-A\Vert _\F\leq\varepsilon)\geq 1-\delta$. 
$\mathcal{C}(A,r,\varepsilon)$ is also not too small: we will show that there exist algorithms aware of $r$ and whose sample complexity nearly matches the lower bound obtained by considering this neighborhood. 

The following theorem provides sample complexity lower bounds in both the multivariate regression and the linear system identification settings.


\begin{theorem} Let $\varepsilon>0$ and $\delta\in (0,1)$. Suppose that $\EE[\hat{\Sigma}] \succ 0$ and  $\mathrm{rank}(A)\le r\leq \underline{d} /2$. 
Then:\\
(i) Multivariate regression. The sample complexity $N$ of any $(\varepsilon,\delta, r)$-stable algorithm in $A$ satisfies:

\begin{equation}\label{eq:lb1}
N\gtrsim  \frac{\sigma^2}{\varepsilon^2}\left(\frac{rd_x+\log (\frac{1}{\delta})}{\bar{\lambda}_r(\Sigma)}\vee\frac{rd_y+\log (\frac{1}{\delta})}{\underline{\lambda}_r(\Sigma)}\right).
\end{equation}

(ii) Linear system identification. Further assume that $\varepsilon\leq \Vert \Gamma_\infty(A)\Vert_2^{-{3}}/12$. The sample complexity $N$ of any $(\varepsilon,\delta, r)$-stable algorithm in $A$ satisfies:

\begin{equation}\label{eq:lb2}
N\gtrsim \sigma^2\frac{rd_x+\log (\frac{1}{\delta})}{\varepsilon^2\underline{\lambda}_r(\Gamma_\infty(A))}.
\end{equation}
\label{thm:full}
\end{theorem}


\vspace{-0.65cm}
\subsection{Lower bounds for rank-adaptive algorithms}\label{subsec:lb_adapt}


We now derive a refined lower bound for rank-adaptive algorithms. To this aim, we define a neighborhood obtained by perturbing the matrix $A$ using matrices whose ranks are smaller than the optimal effective rank. We begin by defining the latter.
Using the results derived for rank-constrained algorithms, we know that for a given selected rank $k$ and given $n$ samples, the minimal Frobenius estimation error  for the estimation of $\Pi_k(A)$ should behave as, in the multivariate case, 
\begin{equation}
\mathrm{ErrReg}(k,n,\Sigma):=\sigma^2\!\left(\frac{kd_x+\log (\frac{1}{\delta})}{n\bar{\lambda}_k(\Sigma)}\vee\frac{kd_y+\log (\frac{1}{\delta})}{n\underline{\lambda}_k(\Sigma)}\right)
\end{equation}
and in the system identification case
\begin{equation}
\mathrm{ErrLti}(k,n,\Gamma_\infty(A)):=\sigma^2\frac{kd_x+\log (\frac{1}{\delta})}{n\underline{\lambda}_k(\Gamma_\infty(A))}.    
\end{equation}
Since we ignore the remaining singular values of $A$, this error square must be increased by $\sum_{i>k}s_i^2(A)$. We can hence guess that the optimal effective rank is $k^\star_{A,n}$ solving in the multivariate case (and similarly in the system identification case)

\begin{equation}
k^\star_{A,n}:=\arg\min_{k\in[r] } \Bigg( \mathrm{ErrReg}(k,n,\Sigma)+\sum_{i>k}s_i^2(A)\Bigg).
\end{equation}

We are now ready to define a notion of stability for rank-adaptive algorithms. 

\begin{definition}[$(\varepsilon,\delta)$-stability]{\it For any integer $m\ge 1$, define ${\cal D}(A,m,\varepsilon) = {\cal C}(A,k^\star_{A,m},\varepsilon)$. An algorithm is $(\varepsilon,\delta)$-stable in $A$ if the following assertion holds:
\begin{equation}\label{eq:ada_stability}
    \exists N\in \mathbb{N} \quad \mathrm{s.t} \quad  \begin{cases}
        (a) \ \Vert A-\Pi_{k^\star_{A,N}}(A)\Vert_F\leq \varepsilon \\
        (b) \ \forall n\geq N, \ \forall A'\in {\cal D}(A,N,\varepsilon),  \\ \quad \mathbb{P}_{ A'}(\Vert \hat A_n-A'\Vert _\F\leq\varepsilon )\geq 1-\delta.
    \end{cases} 
\end{equation}
}
\label{def:full2}
\end{definition}
\vspace{-0.3cm}
If an algorithm is $(\varepsilon,\delta)$-stable in $A$, we define its {\it sample complexity} as the minimal integer $N$ such that Assertion (\ref{eq:ada_stability}) holds. Observe that the definition of $(\varepsilon,\delta)$-stability involves a neighborhood ${\cal D}(A,N,\varepsilon)$ of $A$ that depends on the number of samples $N$. Such a dependence is essential because we wish to analyze algorithms that adaptively select the effective rank as a function of $N$. 
The class of algorithms that are $(\varepsilon,\delta)$-stable in $A$ is broad and includes all consistent algorithms. 
\begin{theorem} Let $\varepsilon>0$ and $\delta\in (0,1)$. Suppose that $\EE[\hat{\Sigma}] \succ 0$ and  $\mathrm{rank}(A)\le r\leq \underline{d} / 2$. 
Then: \\
(i) Multivariate regression. The sample complexity $N$ of any $(\varepsilon,\delta)$-stable algorithm in $A$ satisfies: {$$N\geq \min\left\{n: \gamma_A^\delta(n)\leq \frac{32}{\log(2)} \varepsilon^2\right\},$$} where we define for any $n \ge 1$, 
\begin{align}
\gamma_A^\delta(n): =
\min_{k\in[r]} \Bigg(
  \mathrm{ErrReg}(k,n,\Sigma) + \sum_{i>k} s_i^2(A)
\Bigg).
\end{align}
(ii) Linear system identification. Further assume that $\varepsilon\leq \Vert \Gamma_\infty(A)\Vert_2^{-{3}}/12$. The sample complexity $N$ of any $(\varepsilon,\delta)$-stable algorithm in $A$ satisfies: {$$N\geq \min\{n: \beta_A^\delta(n)\leq  \frac{640}{\log(2)}\varepsilon^2\},$$} where we define for any $n \ge 1$, 
\begin{equation}
\beta_A^\delta(n):=\min_{k\in[r] } \left( \mathrm{ErrLti}(k,n,\Gamma_\infty(A)) +\sum_{i>k}s_i^2(A)\right).
\end{equation}
\label{thm:reduced}
\end{theorem}
\vspace{-0.5cm}
\section{MATRIX DENOISING VIA SINGULAR VALUE THRESHOLDING}\label{sec:denoising}

A key component of our algorithms is an adaptive singular value thresholding procedure that can be combined with any estimator of $A$. Specifically, let $\bar{A}$ be a given estimate of $A$, and denote $Z = \bar{A} - A$ its estimation error. Given a threshold parameter $\xi > 0$, the procedure runs an SVD decomposition to obtain $\bar{A}=\sum_{i=1}^{\underline{d}}s_i(\bar{A})u_iv_i^\top$, 
then produces 
\begin{equation*}
    \bar{A}(\xi)=\sum_{i=1}^{\underline{d}}\mathds{1}_{ \lbrace s_i(\bar A)> \xi \rbrace} s_i(\bar{A}) u_iv_i^\top.
\end{equation*} 
{\bf Singular value thresholding with tighter guarantees.} In his seminal work \cite{chatterjee2015matrix}, Chatterjee proposes a universal way of selecting the threshold $\xi$ and provides a performance analysis of the resulting algorithm. His analysis builds on a crucial result (Lemma 3.5 in \cite{chatterjee2015matrix}) stating that for  $\tau>0$ and $\xi =(1+\tau)\Vert Z\Vert_2$, it must hold that
$
\|\bar{A}\left(\xi\right)-A\Vert_\F^2\leq f(\tau) \Vert Z\Vert_2 \Vert A\Vert_1
$, 
where $f(\tau) = ((4+2\tau)\sqrt{2/\tau}+\sqrt{2+\tau})^2$. This upper bound appears conservative as it depends on the nuclear norm of $A$, and consequently on all the singular values of $A$. Furthermore, observe that when $\tau\rightarrow \infty$ then the upper bound also goes to infinity while it is clear that $\|\bar{A}\left(\xi\right)-A\Vert_\F^2\rightarrow\Vert A\Vert_\F^2$. Ideally, we seek an error upper bound that is independent of the singular values of $A$ exceeding the threshold $\xi$. In the following theorem, we derive an error upper bound satisfying this desired property.
\begin{theorem}
For $\xi \ge 2\| Z\|_2$, we have: 
$$
    \Vert \bar{A}(\xi) -A\Vert_\F^2\leq 18\min_{k\in[r]}\left( 4 k\xi^2+ \sum_{i>k}s_i^2(A)\right).
$$
\label{thm:opt_best_k}
\end{theorem}
\vspace{-0.3cm}
The theorem has a significant implication for the design of rank-adaptive algorithms. Suppose that we can construct an initial estimate $\bar A$ of $A$ and derive a concentration result for its error $Z=\bar{A}-A$. For instance, assume that we can upper bound $2\|Z\|_2$ by $\xi$ with high probability. In this case, a rank-adaptive algorithm that outputs $\bar{A}(\xi)$, i.e., with effective rank $k=\max\{i: s_i(\bar{A})>\xi\}$, is expected to perform well. Theorem \ref{thm:opt_best_k} is a consequence of the following decomposition lemma. 
\begin{lemma}
For $k \in [\underline{d}] $, we have:
\begin{equation*}
    \Vert \Pi_k(\bar A)-A\Vert_\F \leq 2\sqrt{2}\Vert \Pi_{k}(Z)\Vert_\F + 3\Vert A -\Pi_k(A)\Vert_\F.
\end{equation*}
\label{lem:trade_off_k}
\end{lemma}
\vspace{-0.3cm}
We can immediately deduce from this lemma that $\Vert \Pi_k(\bar A)-A\Vert_\F^2\leq 18(k\Vert Z\Vert_2^2+ \sum_{i>k}s_i^2(A))$ for $k\le \underline{d}$. {The proofs of Theorem \ref{thm:opt_best_k} and Lemma \ref{lem:trade_off_k} are given in Appendix \ref{app:sec5}.}


{\textbf{Improved guarantees for estimators with nuclear norm penalization.}  To better appreciate Theorem \ref{thm:opt_best_k}, we provide in Appendix \ref{app:sec5}, as an example, an improved analysis of an estimator with nuclear norm penalization, previously analyzed in \cite{negahban2011estimation}. For this estimator, an adaptive error upper bound was given in Theorem 12 in \cite{bunea2011optimal}. We show in Corollary \ref{coro:nuclear} that our thresholding yields a smaller upper bound, by a multiplicative factor $\kappa(\Sigma)$.}

\section{THE THRESHOLDED LEAST SQUARES ESTIMATOR}\label{sec:lse}

In this section, we describe how Least Squares Estimator (\LSE) can be combined with singular value thresholding to develop \emph{rank-constrained} and \emph{rank-adaptive} algorithms, for both the multivariate regression and linear system identification tasks. We provide performance guarantees for these algorithms. The proofs are deferred to Appendix \ref{app:sec6}. Before proceeding, we define the \LSE ~as $\bar{A} \in \arg\min_{A \in \RR^{d_y \times d_x}} \sum_{i=1}^n \Vert y_i - Ax_i \Vert^2_2$. We can also express it in its closed form as follows $\bar{A} = (Y^\top X) (X^\top X)^{\dagger}$. 
We assume $n \geq d_x$, ensuring a nonzero probability that $X$ is full rank and $A$ is identifiable. This is standard in the literature (e.g \cite{wainwright2019high}), because even for learning a rank-one matrix, observing at least $d_x$ samples is necessary. In the multivariate regression case, we restrict our attention to Gaussian inputs for ease of exposition. However, we prove general versions of Theorems \ref{thm:rlse-gauss} and \ref{thm:tlse-gauss} for any $\hat \Sigma$  in Appendix \ref{app:sec6}, and discuss when one can  replace $\hat \Sigma$ by $\Sigma$ with recent concentration results obtained by \cite{barzilai2024simple}.





\subsection{Multivariate regression}


{\it Rank-constrained LSE} (\RLSE). Here, we consider the case where the rank of $A$ is known to be upper bounded by $r$, thus the algorithm must set the effective rank equal to $r$. We define the \RLSE~as $\hat{A}_n:=\Pi_r(\bar{A})$. 
To analyze the performance of \RLSE~we first establish the concentration result below. 

\begin{lemma}
    Let $\delta \in (0,1)$. For all $k \in [\underline{d}]$,  with probability at least $1-\delta$, it holds that:  
    \begin{align}
    \Vert\Pi_k(Z)\Vert_\F^2\leq \frac{ k\sigma^2\left(\sqrt{d_x}+\sqrt{d_y}+\sqrt{\log\left(\frac{1}{\delta}\right)}\right)^2}{n\underline{\lambda}^H_k(\hat\Sigma)},
    \end{align}  
    where $\underline{\lambda}^H_k(\hat \Sigma):= \big(\frac{1}{k}\sum_{i=d_x-k+1}^{d_x}\frac{1}{\lambda_i(\hat \Sigma)}\big)^{-1}$\! is the harmonic mean of the $k$ smallest eigenvalues of $\hat{\Sigma}$.
    \label{lem:pi_k_lse}
\end{lemma}
Combining this concentration result to Lemma \ref{lem:trade_off_k}, we immediately obtain: 
\begin{theorem}
    Assume that $\mathrm{rank}(A)\leq r$, and $x_i\sim \mathcal{N}(0,\Sigma)$ with $\Sigma\succ 0$. For $n\gtrsim d_x+\log(\frac{1}{\delta})$, the \RLSE~satisfies, with probability at least $1-\delta$:
    
    $$ \Vert \hat{A}_n-A\Vert_\F^2\lesssim r\sigma^2 \frac{\bar d+\log(\frac{1}{\delta})} {n\underline{\lambda}^H_r(\Sigma)}.$$
\label{thm:rlse-gauss}
\end{theorem}
\vspace{-0.5cm}
The upper bound provided above for the \RLSE~is tighter than that of the best-known rank constrained  algorithm found in Corollary 6 of \cite{bunea2011optimal}, by a factor $\underline{\lambda}^H_r({\Sigma}) / \lambda_\mathrm{min}({\Sigma})$. 

Since \RLSE~is consistent and returns a rank-$r$ matrix, it is $(\varepsilon,\delta,r)$-stable for any $\epsilon> 0$ and $\delta \in (0,1)$. Its sample complexity  verifies, see Appendix \ref{app:stability},
\begin{equation}
 N\lesssim 
  r\sigma^2\frac{\bar d+\log(\frac{1}{\delta}) }{\varepsilon^2\underline{\lambda}^H_r( \Sigma)}.    
\end{equation}
While the sample complexity upper bound provided for \RLSE~exhibits a dependence on $r, d_x, d_y, n$ and $\sigma$ matching those identified in the lower bound result of Theorem \ref{thm:full}, it still does not capture the right dependence on $\Sigma$. To fully match the lower bound, the term $\underline{\lambda}_r^H(\Sigma)$ should be replaced by either $\underline{\lambda}_r(\Sigma)$ if $ \bar{\lambda}_r(\Sigma)( \log(1/\delta)+rd_y) \ge \underline{\lambda}_r(\Sigma) (\log(1/\delta)+rd_x)$ or $\bar{\lambda}_r(\Sigma)$ otherwise. However, for well-conditioned $\Sigma$, our bounds are tight. Furthermore, due to our application of Lemma \ref{lem:pi_k_lse}, we obtain $r\log(\frac{1}{\delta})$ compared to $\log(\frac{1}{\delta})$ in our lower bound. We argue that this is negligible in the low-rank settings where $r=O(1)$. 




\begin{table*}[b]
  \centering
  \scalebox{0.8}{
    \begin{minipage}{\textwidth}
    \centering
      \caption{Performance of T-LSE and RSC in presence of alignment between covariate and target matrices. The values correspond to the relative Frobenius error multiplied by $10^{-7}$ for $r=10$ and $10^{-6}$ for $r=45$.}
      \label{table:tlse_vs_rsc}
      \bigskip
      \begin{tabular}{lcccccc}
        \toprule
        & \multicolumn{3}{c}{$r=10$} & \multicolumn{3}{c}{$r=45$} \\
        \cmidrule(lr){2-4} \cmidrule(lr){5-7}
        & Error Avg & Error Std & Error Max & Error Avg & Error Std & Error Max \\
        \midrule
        T-LSE & 1.67 & {\bf 1.03} & {\bf 4.75} & {\bf 5.49} & {\bf 0.69} & {\bf 8.20} \\
        RSC   & {\bf 1.37} & 1.23 & 7.62 & 8.37 & 4.86 & 23.1 \\
        \bottomrule
      \end{tabular}
    \end{minipage}%
  }

\end{table*}

{\it Thresholded LSE} (\TLSE). We now consider rank-adaptive algorithms. 
One can show that the error $Z=\bar{A}-A$ of the LSE satisfies the following concentration result (see Lemma 3 in \cite{bunea2011optimal}): with probability at least $1 - \delta$,
\begin{equation}\label{eq:thresh_LR}
        2\Vert Z\Vert_2\leq\xi_{\textup{MR}} :=  \frac{2\sigma\left(\sqrt{d_x}+\sqrt{d_y}+\sqrt{\log(\frac{1}{\delta})}\right)}{\sqrt{n\lambda_\mathrm{min}(\hat\Sigma)}}.
\end{equation}
Combining this concentration bound with Theorem \ref{thm:opt_best_k}, we obtain the following result. 

\begin{theorem}
    Assume that $x_i\sim \mathcal{N}(0,\Sigma)$ with $\Sigma\succ 0$. For $n\gtrsim d_x+\log(\frac{1}{\delta})$, the \TLSE~ $\hat{A}_n:=\bar{A}(\xi_{\textup{MR}})$ satisfies, with probability at least $1-\delta$:
    $$
\Vert \hat{A}_n-A\Vert_\F^2\lesssim \min_{k\in[r] }\left( k\sigma^2\frac{\bar d+\log(\frac{1}{\delta})}{n\lambda_\mathrm{min}(\Sigma)}+ \sum_{i>k}s_i^2(A)\right).
$$
\label{thm:tlse-gauss}
\end{theorem}
\vspace{-0.3cm}
To the best of our knowledge, \TLSE~enjoys state-of-the-art performance guarantees. It achieves tighter error upper bounds compared to both the rank-adaptive algorithm from \cite{bunea2011optimal}, cf Equation (\ref{eq:soa1}), and the thresholded estimator with nuclear norm penalization introduced in Appendix \ref{app:sec5}. 
More precisely, one can show (see Appendix \ref{app:stability}) that \TLSE~is $(\varepsilon,\delta)$-stable, for any $(\varepsilon,\delta)$, with sample complexity 
\begin{equation}
    N\leq  \min\Bigg\{m: \phi(m)
\leq c\varepsilon^2\Bigg\},
\end{equation}
where 
$$\phi(m)=k^\star_{A,m}\sigma^2\frac{\bar d+\log(\frac{1}{\delta})}{m\lambda_\mathrm{min}(\Sigma)}+ \sum_{i>k^\star_{A,m}}s_i^2(A).$$
$c\in[0,1]$ is a universal constant. 

The error guarantees for \TLSE\ almost match the limits identified in Theorem \ref{thm:reduced}. To achieve an exact match, it suffices to replace, in Theorem \ref{thm:tlse-gauss}, ${\lambda}_{\min}(\Sigma)$ with either $\underline{\lambda}_k(\Sigma)$, if $ \bar{\lambda}_k(\Sigma)  (\log(1/\delta) + kd_y) \geq \underline{\lambda}_k(\Sigma) (\log(1/\delta) + kd_x) $, or $\bar{\lambda}_k(\Sigma)$ otherwise.

\subsection{Linear System Identification}

Analyzing the performance of the LSE in linear system identification requires leveraging recent concentration results on self-normalized processes (see, e.g., the proof of Theorem 3 in \cite{jedra2022finite}). This analysis allows us to upper bound the error $Z$, provided that the number of samples satisfies 
\begin{equation}  n \ge \frac{  \max(c_0\sigma^4,1) \Vert \Gamma_\infty(A) \Vert^3 }{\lambda_{\min}(\Gamma_\infty(A))} \left( \log\left(\frac{1}{\delta}\right) +  d_x\right),
\label{eq:lti_condition}
\end{equation}
for some universal constant $c_0>0$.

{\it Rank-constrained LSE.} Assume first that the algorithm selects an effective rank equal to $r\ge \textrm{rank}(A)$. The \RLSE~ is defined as $\hat{A}_n:=\Pi_r(\bar{A})$. Assuming the noise is Gaussian with variance $\sigma^2$, we can establish the following concentration result for $\Pi_k(Z)$. 

\begin{lemma}
Assume $n$ verifies inequality (\ref{eq:lti_condition}). Then, for all $k \in [\underline{d}]$, with probability at least $1-\delta$:
\begin{enumerate}
    \item The error $Z$ verifies
\begin{equation*}\Vert\Pi_k(Z)\Vert_\F^2\leq k\sigma^2\frac{ d_x+\log(\frac{1}{\delta})}{n\underline{\lambda}^H_k(\hat \Sigma)}.
\end{equation*}
\item $\hat{\Sigma}$ verifies 
\begin{equation*}
    \forall i \in [d_x], \quad  \frac{\lambda_i(\Gamma_{\infty}(A)) }{8}\le \lambda_i(\hat{\Sigma}) \le \frac{3\lambda_i(\Gamma_\infty(A))}{2}. 
\end{equation*}
\end{enumerate}
\label{lem:pi_k_lse_lds}
\end{lemma}
Combining these concentration results with Lemma \ref{lem:trade_off_k}, we obtain:
\begin{theorem}\label{thm:rlse-lti}
Let $\delta \in (0,1)$. Assume that $n$ verifies (\ref{eq:lti_condition}) and that $\mathrm{rank}(A)\leq r$. The \RLSE~satisfies, with probability at least $1-\delta$:

\begin{equation*}
    \Vert \hat{A}_n-A\Vert_\F^2\lesssim r\sigma^2\frac{ d_x+\log(\frac{1}{\delta})}{n\underline{\lambda}^H_r( \Gamma_\infty(A))}.
\end{equation*}  
\end{theorem}

{\it Thresholded LSE.} Consider now rank-adaptive algorithms. When the noise is Gaussian with variance $\sigma^2$ and when $n$ verifies inequality (\ref{eq:lti_condition}), we establish (refer to the proof of Lemma \ref{lem:pi_k_lse_lds}) that, with probability at least $1-\delta$: 
\begin{equation}\label{eq:thresh_LDS}
    2\Vert Z\Vert_2 \leq \xi_{\textup{SysID}} := 2\sigma\sqrt{\frac{d_x+\log(\frac{1}{\delta})}{n\lambda_\mathrm{min} (\hat\Sigma)}}. 
\end{equation}
Combining the previous result and Theorem \ref{thm:opt_best_k}, we obtain the following result. 
\begin{theorem}\label{thm:tlse-lti}
    Assume $n$ verifies inequality (\ref{eq:lti_condition}). The \TLSE~satisfies, with probability at least $1-\delta$:  $$
\Vert \hat{A}_n-A\Vert_\F^2\lesssim \min_{k \in [r]}\left( \frac{k\sigma^2(d_x+\log(\frac{1}{\delta}))}{n\lambda_\mathrm{min} (\Gamma_\infty(A))}+ \sum_{i>k}s_i^2(A)\right).
$$
\label{thm:56}
\end{theorem}

To the best of our knowledge, this is the first rank-adaptive performance bound for system identification. Furthermore, the same remarks as in the multivariate regression case can be made regarding the remaining existing gaps between our lower and upper bounds.


    

\begin{figure*}[t]
    \includegraphics[width=0.485\textwidth]{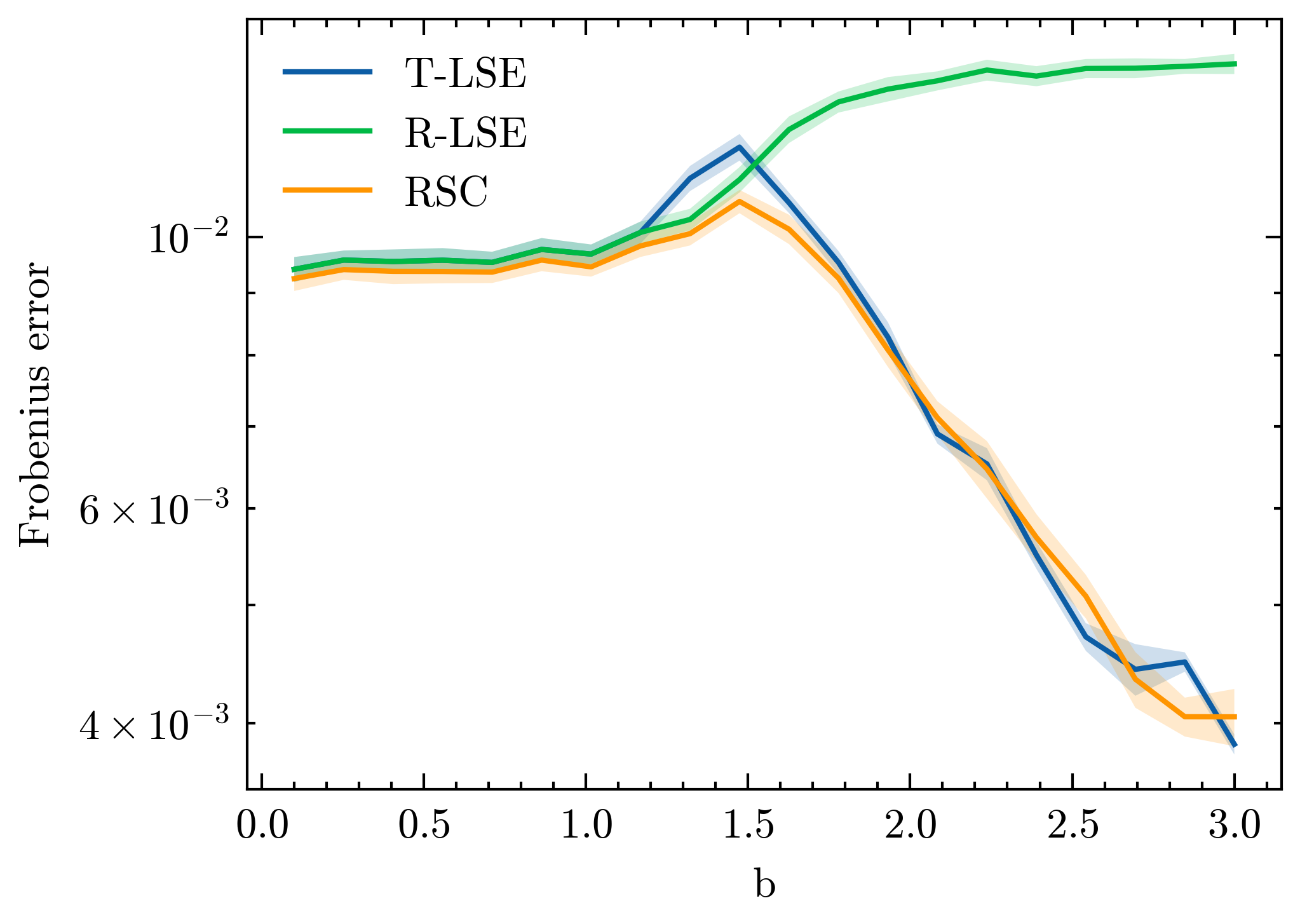}
    \hspace{0.3cm}
    \includegraphics[width=0.45\textwidth]{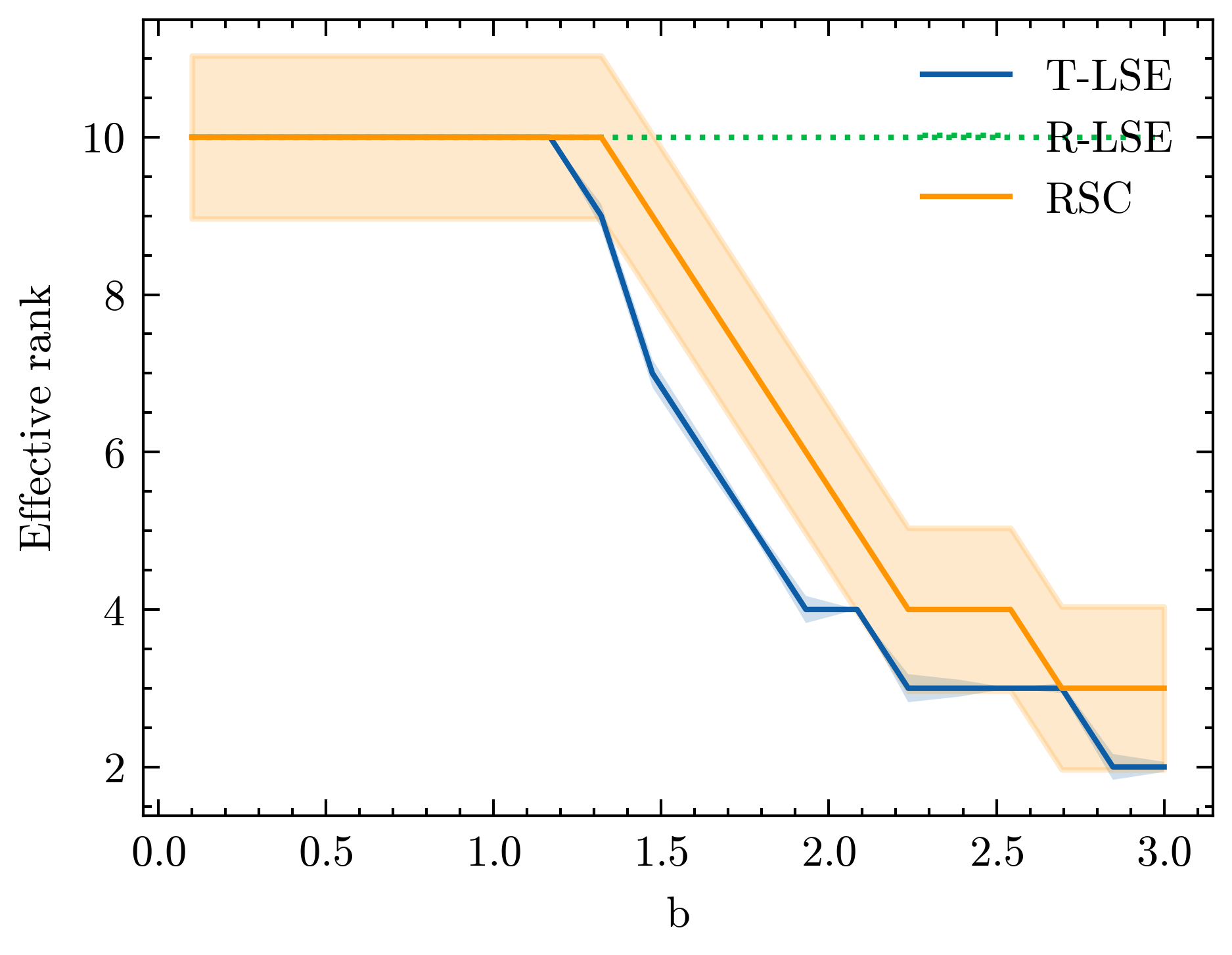}
    \vspace*{-3mm}
    \caption{Frobenius error (left) and effective rank (right) vs noise level $b$, $(r=10)$.}
    \label{fig:fro_rank_low}
    
\end{figure*}
\section{NUMERICAL EXPERIMENTS}\label{sec:exp}

We test our algorithms experimentally in the multivariate regression case using similar settings as existing work \cite{bunea2011optimal}. We estimate $A\in\mathbb{R}^{d\times d}$ of rank $r$ with $d=50$ using $n=1000$ samples. Our experiments can also be conducted for higher values of $(d,r,n)$, and these can be found in Appendix \ref{app:sec7}. Each 'run' involves  sampling the design matrix $X$ and the Gaussian noise matrix $E$ to compute $Y=XA+E$. We evaluate the performance of an estimator $\hat A_n$ using the relative Frobenius error $ \Vert \hat A_n
- A\Vert_{\textup{F}}^2 / \Vert A\Vert_{\textup{F}}^2$ averaged over $T=30$ runs. We consider two settings: (1) Low-rank $r=10, \sigma=0.1$. (2) High-rank $r=45, \sigma=0.4$.


{\bf On the impact of singular subspace alignment.} We begin by comparing \TLSE~ and the RSC algorithm \cite{bunea2011optimal} (initially designed to minimize $\Vert X\hat{A}_n-XA\Vert_\F$), and in particular show how the 'alignment' of the singular subspaces of the design matrix $X$ and of $A$ impact the performance of these two algorithms. To this aim, we start by sampling $X$: for $i=1,...,n, \quad x_i\sim \mathcal{N}(0,\Sigma)$ where $\Sigma$ is a diagonal matrix such that $\Sigma_{j,j}=j^2$ for $j=1,...,d$. Let $X=U_XS_XV_X^\top$ be its SVD. Then we create a set of $d+1$ matrices $A$ with different levels of alignment, as follows: (1) Sample entries of $A$ uniformly at random in $[0,1]$, and let $A=USV^\top$ be its SVD. (2) Change the singular value of $A$ as $s_j(A)=\frac{1}{j^{b}}$ for $j=1,...,r$ (we use $b=1.5$ in the low-rank case and $b=0.5$ in the high-rank case). Note that the assumption of  decaying singular values is frequent, see \cite{wu2020adaptive}. (3) Change its right singular vectors to align them to those of $X$: $A=USPV_X^\top$, where $P$ is one of the $(d+1)$ circular shifting matrices (i.e., shifting circularly the positions of the rows $V_X^\top$). We have constructed $d+1$ scenarios with different levels of alignment, and report the performance of \TLSE\ and RSC in Table~\ref{table:tlse_vs_rsc}, averaged across all scenarios as well as for the worst-case alignment. \TLSE\ demonstrates greater robustness and consistently outperforms RSC, particularly in challenging alignment conditions.

{\bf On the importance of adaptivity.} Next we compare adaptive algorithms, \TLSE\ and RSC, to the non-adaptive algorithm \RLSE\ aware of the true rank of $A$, $r=10$. If the number of samples $n$ is fixed and if the noise level increases, we expect that optimal algorithms should target an effective rank strictly smaller than $r$. Here we increase the level of noise by increasing a parameter $b$. The singular values of $A$ are now set to $s_j(A)=\frac{1}{j^b}$ for $j=1,...,r$, and we do not align its singular vectors with those of $X$. In Figure \ref{fig:fro_rank_low}, for $r=10$, we observe two interesting regimes. When $b\in [0.1,1.5]$ is small, the signal is still strong enough for \RLSE\ to be competitive compared to the rank-adaptive algorithms. For $b\in[1.5,3]$ both adaptive algorithms \TLSE\ and RSC target a lower effective rank and have better performance than \RLSE. Additional figures are provided in Appendix \ref{app:sec7} for the high-rank regime.

\section{CONCLUSION}

We revisited the problem of estimating high-dimensional matrices via reduced-rank regression, focusing on rank-adaptive algorithms. These methods estimate the matrix by learning its singular values and corresponding singular vectors up to an effective rank, which is adaptively determined from the data. For this setting, we established  instance-specific lower bounds on the sample complexity that any such algorithm must satisfy. These bounds explicitly depend on the spectral properties of the underlying matrix and the covariance structure of the covariates, offering insight into the effective rank that an optimal algorithm should infer.

We proposed the thresholded \LSE~, which combines the classical \LSE~ with a universal singular value thresholding procedure. We derived finite-sample error bounds for this estimator and showed that it achieves better perfomance guarantees than existing algorithms. A small gap remains between the performance of \LSE\ and the theoretical limits. Closing this gap is an interesting direction for future research. Another promising avenue is the development of similarly adaptive estimation procedures under stronger norms—such as entry-wise norms—which are particularly relevant in applications like reinforcement learning, where such guarantees are often required.

\paragraph{Acknowledgments.}

This research is supported by Vetenskapr{\aa}det, Digital Futures, and the Wallenberg AI, Autonomous Systems and Software program.
%

\bibliography{bibliography}
\newpage
\section*{Checklist}



\begin{enumerate}

  \item For all models and algorithms presented, check if you include:
  \begin{enumerate}
    \item A clear description of the mathematical setting, assumptions, algorithm, and/or model. \textbf{Yes. Sections 2 and 6 define the sampling model and least square estimator respectively.}
    \item An analysis of the properties and complexity (time, space, sample size) of any algorithm. \textbf{Yes. Sections 4 and 6 presents an analysis of the sample complexity with lower/upper bounds. Section 5 provides with additional properties for thresholded estimators.}
    \item (Optional) Anonymized source code, with specification of all dependencies, including external libraries. \textbf{Yes. An anonymous Jupyter notebook with all required dependencies to run our experiments is provided.}
  \end{enumerate}

  \item For any theoretical claim, check if you include:
  \begin{enumerate}
    \item Statements of the full set of assumptions of all theoretical results. \textbf{Yes. Section 4 presents our lower bounds with assumptions on expected covariance matrix and the rank. Section 6 presents concentration results with the required sample complexity.}
    \item Complete proofs of all theoretical results. \textbf{Yes. All our theoretical results have corresponding proof in the supplementary material, see Appendices  \ref{app:sec4}, \ref{app:sec5}, \ref{app:sec6}, \ref{app:stability}.}
    \item Clear explanations of any assumptions. \textbf{Yes. Our assumptions are standard for our objective i.e identifiability (positive definiteness of $\Sigma$), low-rankness ($2r\leq d$), empirical covariance concentration.}     
  \end{enumerate}

  \item For all figures and tables that present empirical results, check if you include:
  \begin{enumerate}
    \item The code, data, and instructions needed to reproduce the main experimental results (either in the supplemental material or as a URL). \textbf{Yes. The Jupyter notebook in supplementary material is structured into several headings to clearly reproduce experiments.}
    \item All the training details (e.g., data splits, hyperparameters, how they were chosen). \textbf{Yes. We present choices of hyperparameters and their influence in Appendix \ref{app:sec7}.}
    \item A clear definition of the specific measure or statistics and error bars (e.g., with respect to the random seed after running experiments multiple times). \textbf{Yes. We define the Frobenius error and rank of our estimators either on Y-axis of a figure or as numerical value in a table.}
    \item A description of the computing infrastructure used. (e.g., type of GPUs, internal cluster, or cloud provider). \textbf{Yes. See Appendix \ref{app:resources}.}
  \end{enumerate}

  \item If you are using existing assets (e.g., code, data, models) or curating/releasing new assets, check if you include:
  \begin{enumerate}
    \item Citations of the creator If your work uses existing assets. \textbf{Yes. We used a visualisation library whose repository is cited in Appendix \ref{app:resources}.}
    \item The license information of the assets, if applicable. \textbf{Not Applicable.}
    \item New assets either in the supplemental material or as a URL, if applicable. \textbf{Not Applicable.}
    \item Information about consent from data providers/curators. \textbf{Not Applicable.}
    \item Discussion of sensible content if applicable, e.g., personally identifiable information or offensive content. \textbf{Not Applicable.}
  \end{enumerate}

  \item If you used crowdsourcing or conducted research with human subjects, check if you include:
  \begin{enumerate}
    \item The full text of instructions given to participants and screenshots. \textbf{Not Applicable.}
    \item Descriptions of potential participant risks, with links to Institutional Review Board (IRB) approvals if applicable. \textbf{Not Applicable.}
    \item The estimated hourly wage paid to participants and the total amount spent on participant compensation. \textbf{Not Applicable.}
  \end{enumerate}

\end{enumerate}

\onecolumn
\newpage\tableofcontents
\newpage
\appendix
\section{Proofs of results presented in Section \ref{sec:lb}}\label{app:sec4}

In this section, we provide the proofs for our complexity sample lower bounds, namely Theorem \ref{thm:full} and Theorem \ref{thm:reduced}. The proofs rely on intricate change-of-measure arguments that allow for multiple hypotheses.  In \textsection\ref{app:sec4:tools}, we start by presenting some of the notations and tools used in these arguments. In  \textsection\ref{app:sec4:packing}, we prove a result on the packing number of the Stiefel manifold given in Lemma \ref{lem:stiefel}. We then present the proofs of Theorem \ref{thm:full} and Theorem \ref{thm:reduced} in \textsection\ref{app:sec4:thm1} and  \textsection\ref{app:sec4:thm2}, respectively.

\subsection{Tools for the change-of-measure argument}\label{app:sec4:tools}

We start by introducing notation which we will be used throughout this section. We consider a collection of $n$ observations $(x_1,y_1), \dots, (x_n, y_n)$. We will often  view these observations as samples  from a probability distribution with a density function $f_A$ parameterized by $A \in \RR^{d_y \times d_x}$. For some given $A \in \RR^{d_y\times d_x}$, we use $\EE_A$ (resp. $\PP_A)$ to denote the expectation (resp. the probability measure) under the distribution with the density $f_A$.

\paragraph{Computation of the expected log-likelihood ratio.}\label{sec:log-likelihood-ratio}


Let $A, A' \in \RR^{d_y \times d_x}$ such that $A \neq A'$. We define the log-likelihood ratio of a collection of $n$ observations $(x_1,y_1), \dots, (x_n, y_n)$, under $A$ and $A'$ as follows: 

$$
L_n(A, A')=\log\left(\frac{f_{A}((x_1,y_1), \dots, (x_n, y_n))}{f_{A'}((x_1, y_1), \dots, (x_n, y_n))} \right).
$$

When $(x_1,y_1), \dots, (x_n, y_n)$ are generated under a linear model as presented in \textsection \ref{sec:prelim} with Gaussian noise then the expected log-likelihood can be computed explicitly:
\begin{lemma}\label{lem:log-lik}


    Assume that $(x_1, y_1), \dots, (x_n, y_n)$ are generated under a linear model as presented in \textsection\ref{sec:prelim}, i.e., $y_i = Ax_i + \eta_i$, where $(\eta_i)_{i \ge 1}$ is a sequence of i.i.d. of random variables distributed as $\cN(0, \sigma^2 I_{d_y})$. Then, the expected log-likelihood ratio under $A$ and $A'$ is as follows: 
    
    $$
    \mathbb{E}_{A}[L_n(A, A')]=\frac{1}{2\sigma^2}\mathrm{Tr}\left((A-A')^\top(A-A')\mathbb{E}_A[X^\top X]\right).
    $$
\end{lemma}

The computations leading to the above result follow the same steps as in \cite{jedra2019}, we provide a proof for completeness.
\begin{proof}[Proof of Lemma \ref{lem:log-lik}]
Consider that we have $n$ observations $(x_1,y_1), \dots , (x_n, y_n)$. Before computing the expectation of the log-likelihood ratio, we use conditional independence between the observations to write the ratio of densities as a product of ratios. 
In the multivariate regression setting, we have by independence

    \begin{align*}
        \frac{f_{A}((x_i,y_i)_{i=1}^n)}{f_{A'}((x_i,y_i)_{i=1}^n)}&=\prod_{i=1}^n \frac{f_{A}(x_i,y_i)}{f_{A'}(x_i,y_i)} =\prod_{i=1}^n \frac{f_{A}(y_i|x_i)f(x_i)}{f_{A'}(y_i|x_i)f(x_i)} =\prod_{i=1}^n \frac{f_{A}(y_i|x_i)}{f_{A'}(y_i|x_i)}.
    \end{align*}
    
In the system identification case, the observations $((x_i,y_i)_{i=1}^n)$ simplify to $(x_i)_{i=1}^{n+1}$. In a similar fashion, by conditional independence, we have

    \begin{align*}
        \frac{f_{A}(x_1,...,x_{n+1})}{f_{A'}(x_1,...,x_{n+1})}&=\prod_{i=1}^n \frac{f_{A}(x_{i+1}|x_i)}{f_{A'}(x_{i+1}|x_i)}=\prod_{i=1}^n \frac{f_{A}(y_i|x_i)}{f_{A'}(y_i|x_i)}
    \end{align*}
    
with $y_i=x_{i+1}$. 

Therefore,
    

\begin{align*}
        \mathbb{E}_{A}[L_n(A, A')]      &=\mathbb{E}_A\sum_{i=1}^n\left[\log \frac{f_{A}(y_i|x_i)}{f_{A'}(y_i|x_i)}\right] =\sum_{i=1}^n \mathbb{E}_A\left[\log \frac{f_{A}(y_i|x_i)}{f_{A'}(y_i|x_i)} \right] \\
        & =\sum_{i=1}^n \mathbb{E}_A\left[\mathbb{E}_A\left[\log \frac{f_{A}(y_i|x_i)}{f_{A'}(y_i|x_i)} | x_i\right]\right].
\end{align*}
Noting that  $f_{A}(y_i|x_i)=\mathcal{N}(A x_i, \sigma^2I_{d_y})$ (and similarly for $A'$) we obtain
\begin{align*}
        \mathbb{E}_{A}[L_n(A, A')] &= \sum_{i=1}^n \mathbb{E}_A\left[\mathrm{KL}\left(\mathcal{N}(A x_i, \sigma^2I_{d_y}), \mathcal{N}(A' x_i, \sigma^2I_{d_y})\right)\right]\\
        &=  \sum_{i=1}^n\mathbb{E}_A\left[\frac{1}{2\sigma^2} \mathrm{Tr}(x_i^\top(A-A')^\top (A-A')x_i)\right] \\
        &= \frac{1}{2\sigma^2}\sum_{i=1}^n\mathbb{E}_A\left[\mathrm{Tr}((A-A')^\top (A-A')x_ix_i^\top)\right]\\
        &= \frac{1}{2\sigma^2}\mathbb{E}_A\left[\mathrm{Tr}((A-A')^\top (A-A')X^\top X)\right]\\
        &= \frac{1}{2\sigma^2}\mathrm{Tr}\left((A-A')^\top(A-A')\mathbb{E}_A[X^\top X]\right).
\end{align*}
\end{proof}

\paragraph{Data-processing inequality for multiple hypotheses.} We present a variant of the data processing inequality which allows for using multiple hypothesis in deriving lower bounds. This result is borrowed from  \cite{jedra2022finite} (see their Proposition 2).

\begin{lemma}
Let $kl(p,q)$ the \textrm{KL}-divergence between two Bernoulli with parameters $p$ and $q$. Let $n,m$ be two positive integers and consider $\mathcal{E}_1,...,\mathcal{E}_m$ disjoint events belonging to the filtration $\mathcal{F}_n$. Then, for all $A,A_1,...,A_m$ such that for $i=1,...,m, A_i\neq A$, we have the following 

$$\frac{1}{m}\sum_{i=1}^{m}\mathbb{E}_{A_i}(L_n(A_i,A))\geq kl\left(\frac{1}{m}\sum_{i=1}^m \mathbb{P}_{A_i}(\mathcal{E}_i), \frac{\mathbb{P}_A(\cup_{i=1}^m \mathcal{E}_i)}{m} \right).$$

Furthermore, if $\mathbb{P}_{A_i}(\mathcal{E}_i)\geq 1-\delta$ for all $ i \in  [m]$, and  $\cup_{i=1}^m\mathcal{E}_i \subseteq \mathcal{E}^c$ for some event $\mathcal{E}$ such that $\mathbb{P}_{A}(\mathcal{E})\geq 1-\delta$ then 

\begin{align*}
\frac{1}{m}\sum_{i=1}^{m}\mathbb{E}_{A_i}(L_n(A_i,A)) \geq kl\left(1-\delta, \frac{\delta}{m}\right) \geq\frac{1}{2}\log\left(\frac{m}{4\delta}\right).
\end{align*}
\label{lemma:data_process}
\end{lemma}

\begin{remark} The first statement of Lemma \ref{lemma:data_process} resembles Fano's inequality, which states that for all events $\mathcal{E}_i$ and probability measures $\mathbb{P}_i, \mathbb{Q}_i$ with $i\in[m]$, one has

    \begin{equation}
        \frac{1}{m}\sum_{i=1}^m\mathbb{P}_i(\mathcal{E}_i)\leq\frac{\frac{1}{m}\sum_{i=1}^m \mathrm{KL}(\mathbb{P}_i, \mathbb{Q}_i)+\log(2)}{\log m}.
        \label{eq:fano}
    \end{equation} 

    More precisely, both results are derived from the elementary data-processing inequality which involves analyzing $\mathrm{kl}\left(\frac{1}{m}\sum_{i=1}^m\mathbb{P}_i(\mathcal{E}_i),\frac{1}{m}\sum_{i=1}^m\mathbb{Q}_i(\mathcal{E}_i)\right)$. While this quantity can be additionally simplified leading to Fano's inequality, see Proposition 4 of \cite{gerchinovitz2020fano} for technical details, we instead leverage it with our stability definitions. This allows us to derive a tight dependency in $\log(\frac{1}{\delta})$ by carefully choosing the events $\mathcal{E}_i$, as  shown by the second statement of Lemma \ref{lemma:data_process}. 

    On the other hand, Fano's inequality cannot recover alone this $\log(\frac{1}{\delta})$ dependency. As a remedy, authors of \cite{ma2024high} propose a separate LeCam's two point method (see their Corollary 6). However as its name indicates and as far as we are aware, it can only be applied to two hypothesis, and therefore cannot leverage the packing argument from which we derive the correct dependency of $\log(\frac{1}{\delta})$ with respect to all model parameters in our lower bounds.  
    \end{remark}

\subsection{Proof of Lemma \ref{lem:stiefel}} \label{app:sec4:packing}

\begin{proof} 
Let $k\le \frac{d}{2}$. Let $\Gamma$ be the usual Gamma function. We first introduce the geodesic distance on the Stiefel manifold, following the formalism of \cite{henkel2005sphere}. For any skew-symmetric $A\in\mathbb{R}^{k\times k}$ (i.e., such that $A^\top=-A$), and $ B\in\mathbb{R}^{(d-k)\times k}$, we define

$$
X= \begin{pmatrix}
    A & -B^\top \\
    B & 0  
\end{pmatrix} \in \RR^{d\times d}.
$$

The geodesic distance between the matrix $Q=\exp(X)\begin{pmatrix}
    I_k \\
    0
\end{pmatrix}\in \mathrm{St}_k^d(\mathbb{R})$ and $I_{d,k}=\begin{pmatrix}
        I_k \\
        0
\end{pmatrix}\in \mathrm{St}_k^d(\mathbb{R})$ (where  $I_k\in\RR^{k\times k}$ is the identity matrix) is defined by $d_g(Q,I_{d,k})=\frac{1}{2}\Vert X\Vert_\F ^2=\frac{1}{2}\Vert A\Vert_\F^2+\Vert B\Vert_F^2$. The geodesic distance $d_g(Q_1,Q_2)$ between arbitrary  $Q_1,Q_2\in \mathrm{St}_k^d(\mathbb{R})$ follows from the isometric transformation $Q'_1=Q_2^TQ_1, \ Q'_2=I_{d,k}$: it corresponds to $d_g(Q_1',I_{d,k})$.
    
For any $m\geq 1$, Theorem 4.1 in \cite{henkel2005sphere} states the existence of a packing $\mathcal{P}_{k,d}$ of $\mathrm{St}_k^d(\mathbb{R})$ of size $m$, with minimal geodesic distance between its elements
    
    $$\forall (Q_1,Q_2)\in \mathcal{P}_{k,d} \quad d_g(Q_1,Q_2)\geq d_0:= ab.$$
    
    where 
    
    $$a=\left(\frac{1}{2}\right)^\frac{\log_2 m}{2dk-k^2}, \qquad 
b=\Bigg((2\sqrt{\pi})^k \frac{\Gamma\left(\frac{k(2d-k)}{2} + 1\right)}{\prod_{i=d-k+1}^d \Gamma(i)}\Bigg)^\frac{1}{2dk-k^2}
.$$ 

We will prove that with an appropriate choice of $m$, we obtain a packing with the properties required in Lemma \ref{lem:stiefel}. Let us select $m$ such that  $\log_2(m)=2dk-k^2=k(2d-k)$. With this choice, we can see that $a=\frac{1}{2}$. The next lemma, proved below, gives a lower bound on $b$.

\begin{lemma}\label{lem:b}
There exists a universal constant $c\approx 0.17$ such that $b\geq e^{\frac{-1-\log(2)}{3}-\frac{1}{2}-c}\sqrt{\frac{k}{2}}$.
\end{lemma}

From this lemma, we deduce that the distance between two elements of the packing satisfies:
$$d_0=ab\geq  \frac{e^{\frac{-1-\log(2)}{3}-\frac{1}{2}-c}}{2\sqrt{2}}\sqrt{k}.$$
Hence we have proved that there exists a packing of size $2^{k(2d-k)}\geq 2^{kd}$ with minimal geodesic distance lower bounded by $\sqrt{k}$ up to a universal constant. To complete the proof of Lemma \ref{lem:stiefel}, we have to relate the geodesic distance to that induced by the Frobenius norm. To this aim, we use Corollary 7.2 of \cite{mataigne2024bounds} stating that 
    $$\forall (Q_1,Q_2)\in \mathrm{St}_k^d(\mathbb{R}), \quad  d_g(Q_1,Q_2) \leq \frac{\pi}{2}\Vert Q_1-Q_2\Vert_\F.$$
    
    Hence the packing in the geodesic distance also induces a packing in the Frobenius distance with minimal distance $\frac{\sqrt{k}}{C}$ where $C:=\pi\sqrt{2} e^{\frac{1+\log(2)}{3}+\frac{1}{2}+c}$. 

\end{proof}

\begin{proof}[Proof of Lemma \ref{lem:b}]
We use Theorem 1.6 in \cite{batir2008inequalities}: For any $x\geq 1$

$$\left(\frac{x}{e}\right)^x\sqrt{2\pi x}\leq\Gamma(x+1)\leq\left(\frac{x}{e}\right)^x\sqrt{2\pi(x+c)}$$

where $c\approx0.17$. 

Let $x=\frac{\log_2 m }{2}$, then we can use this result to bound $b$ as follows
\begin{align*}
    b 
    &\geq \left((2\sqrt{\pi})^k\frac{(\frac{x}{e})^x\sqrt{2\pi x}}{\prod_{i=d-k}^{d-1} (\frac{i}{e})^i\sqrt{2\pi i}}\right)^\frac{1}{2dk-k^2}\ \left(\prod_{i=d-k}^{d-1} \sqrt{1+\frac{c}{i}}\right)^\frac{-1}{2dk-k^2}\\
    &=\sqrt{\frac{x}{e}} \left(\frac{\sqrt{2}^{k+1}\sqrt{\pi}}{\prod_{i=d-k}^{d-1} (\frac{i}{e})^i\sqrt{ i}}\right)^\frac{1}{2dk-k^2}\ \left(\prod_{i=d-k}^{d-1} \sqrt{1+\frac{c}{i}}\right)^\frac{-1}{2dk-k^2} \\
    &\geq \sqrt{\frac{x}{e}} \underbrace{\left(\frac{1}{\prod_{i=d-k}^{d-1} (\frac{i}{e})^i\sqrt{ i}}\right)^\frac{1}{2dk-k^2}}_{:=b_0}\ \underbrace{\left(\prod_{i=d-k}^{d-1} \sqrt{1+\frac{c}{i}}\right)^\frac{-1}{2dk-k^2}}_{:=b_1},
\end{align*}
where the second line holds since $(\frac{x}{e}) ^\frac{x}{2dk-k^2}=\sqrt{\frac{x}{e}}$, and  $\sqrt{x}^\frac{1}{2dk-k^2}=x^\frac{1}{x}\geq 1$. Next we compute the logarithms of $b_0$ and $b_1$. 

\underline{\it Computing $\log (b_1)$.} We have:
\begin{align*}
    \log\left(\prod_{i=d-k}^{d-1} \sqrt{1+\frac{c}{i}}\right) &= \frac{1}{2}\sum _{i=d-k}^{d-1}\log\left(1+\frac{c}{i}\right) \leq \frac{1}{2}k\log\left(1+\frac{c}{d-k}\right)\\
    &\leq \frac{kc}{2(d-k)}\leq \frac{c}{2} \qquad \textrm{since $k\leq \frac{d}{2}$.}
\end{align*}
We conclude that:
$b_1=\left(\prod_{i=d-k}^{d-1} \sqrt{1+\frac{c}{i}}\right)^\frac{-1}{2dk-k^2}\geq e^{-\frac{c}{2(2dk-k^2)}}\geq e^{-c}.$

\underline{\it Computing $\log (b_0)$.} We decompose the product $\prod_{i=d-k}^{d-1} (\frac{i}{e})^i\sqrt{ i}$ into two products $\prod_{i=d-k}^{d-1}\sqrt{i}$ and $\prod_{i=d-k}^{d-1}(\frac{i}{e})^i$.

\begin{itemize}
    \item We upper bound the logarithm of the first product as follows: 
    
    $
    \log \left(\prod_{i=d-k}^{d-1}\sqrt{i}\right)=\frac{1}{2}\sum_{i=d-k}^{d-1}\log(i)
    \leq \frac{1}{2}k\log(d).$    
We obtain 
$$\left(\prod_{i=d-k}^{d-1}\sqrt{i}\right)^\frac{-1}{2dk-k^2}\geq e^{-\frac{k\log(d)}{2k(2d-k)}}\geq e^{-\frac{\log(d)}{2(2d-k)}}\geq e^{-\frac{1}{3}}.$$
\item For the second product, its logarithm can also be upper bounded as follows:
$\log\left(\prod_{i=d-k}^{d-1} \left(\frac{i}{e}\right)^i\right)=\sum_{i=d-k}^{d-1}i\log(i)-\sum_{i=d-k}^{d-1}i \leq\sum_{i=d-k}^{d-1}i\log(i).$ Since $t\log(t)$ is an increasing convex function for $t\geq 1$, we get

\begin{align*}
    \sum_{i=d-k}^{d-1}i\log(i)&\leq \int_{i=d-k}^dt\log(t)dt\leq k\frac{(d-k)\log(d-k)+d\log(d)}{2}\\
    &= k\frac{(d-k)\log(d-k)+d(\log(d-k)+\log(1+\frac{k}{d-k}))}{2}\\
    &\leq k\frac{(2d-k)\log(d-k)+d\log 2}{2}.
\end{align*}
We deduce that  $\left(\prod_{i=d-k}^{d-1} \left(\frac{i}{e}\right)^i\right)^\frac{-1}{2dk-k^2}\geq e^{-\frac{\log(d-k)}{2}-\frac{d\log(2)}{2(2d-k)}}
    \geq \frac{1}{\sqrt{d-k}}e^\frac{-\log(2)}{3}$.
\end{itemize}
Plugging the above inequalities together, we obtain $b_0\geq e^\frac{-1-\log(2)}{3} \sqrt{\frac{1}{d-k}}.$ \\
We conclude that 
\begin{align*}
b=\sqrt{\frac{x}{e}}b_0b_1
&\geq  e^{\frac{-1-\log(2)}{3}-\frac{1}{2}-c} \sqrt{\frac{x}{d-k}} =e^{\frac{-1-\log(2)}{3}-\frac{1}{2}-c}\sqrt{\frac{k(2d-k)}{2(d-k)}}
\\&\geq e^{\frac{-1-\log(2)}{3}-\frac{1}{2}-c}\sqrt{\frac{k}{2}}.
\end{align*}
\end{proof}

\subsection{Proof of Theorem \ref{thm:full}}\label{app:sec4:thm1} 
\begin{proof}

Let $\varepsilon>0$ and $\delta\in(0,1).$ We suppose that $\EE(\hat \Sigma)\succ 0$. Let $A$ be such that $ \mathrm{rank}(A)\leq r \leq \frac{1}{2}\underline{d}$. Let $\hat A_n$ be an $(\varepsilon,\delta,r)$-stable algorithm in $A$ and let $N$ be its sample complexity. The lower bound on $N$ stated in Theorem \ref{thm:full} is the maximum over two lower bounds, one that depends on $d_x$ and one that depends on $d_y$. We prove both lower bounds below. 

\textbf{1) Lower bound that depends on $d_y$.} We start by applying Lemmas \ref{lem:log-lik} and \ref{lemma:data_process} to $A$ using the confusing models  from ${\cal C}_{1}(A,r,\varepsilon)$. More precisely, introduce for all $i\in[2^{r d_y}]$, 

$$
A_i=A+\frac{2C\varepsilon}{\sqrt r} Q_iW_{-r}^\top , \quad \text{where }\ Q_i\in \mathcal{P}_{r}^{d_y}, 
$$

and $C\geq 1$ is a universal constant previously defined in Lemma \ref{lem:stiefel}.

\textbf{Application of Lemma \ref{lem:log-lik}}: We have

$$\forall i\in[2^{rd_y}], \quad \EE_{A_i} [L_N(A_i,A)] =  \frac{1}{2\sigma^2}\mathrm{Tr}\left((A-A_i)^\top(A-A_i)\mathbb{E}_{A_i}[X^\top X]\right),$$

where $X\in \RR^{N\times d_x}$ is the (random) matrix of covariates.

\begin{itemize}
    \item In the multivariate regression case, we clearly have: $\mathbb{E}_{A_i}[X^\top X]=\mathbb{E}[X^\top X]=N\Sigma,$ and therefore,
    
    \begin{align*}\EE_{A_i} [L_N(A_i,A)] &=  \frac{N}{2\sigma^2}\mathrm{Tr}((A-A_i)^\top(A-A_i)\Sigma) =\frac{4NC^2\varepsilon^2}{2\sigma^2r}\mathrm{Tr}(Q_iW_{-r}^\top \Sigma W_{-r}Q_i^\top) \\
    &=\frac{2NC^2\varepsilon^2}{\sigma^2r}\mathrm{Tr}(W_{-r}^\top \Sigma W_{-r}) = \frac{2NC^2\varepsilon^2 \underline{\lambda}_r(\Sigma)}{\sigma^2}. 
    \end{align*}

    \item In the system identification case, the trajectory is generated by $A$ itself. Lemma 8 in \cite{jedra2022finite} shows that 
    
    \begin{align*}
        \mathbb{E}_{A_i}\left[\sum_{n=0}^{N-2} x_n x_n^\top \right] & = \sum_{n=0}^{N-2} \Gamma_{n}(A_i) \preceq N \Gamma_\infty (A) + N \left(\Gamma_{\infty}(A_i) - \Gamma_{\infty}(A)\right).
    \end{align*}
    
    Furthermore, since $\Vert A_i-A\Vert_2=\frac{2C\varepsilon}{\sqrt{r}}\leq \frac{\Vert \Gamma_\infty(A)\Vert_2^{-3}}{4}$ then Lemma 1 in \cite{jedra2022finite} implies that 
    
    $$\Vert \Gamma_\infty(A_i)-\Gamma_\infty(A)\Vert_2\leq 16\Vert A_i-A\Vert_2\Vert \Gamma_\infty(A)\Vert_2^3.$$ 
    
    Hence 
    \begin{align*}
        &\mathrm{Tr}(W_{-r}^\top\mathbb{E}_{A_i}[X^\top X] W_{-r}) = \sum_{n=0}^{N-2}\mathrm{Tr}(W_{-r}^\top\Gamma_n(A_i) W_{-r})\\
        &\leq N\mathrm{Tr}(W_{-r}^\top\Gamma_\infty(A) W_{-r}) + N\mathrm{Tr}\bigg(W_{-r}^\top\big(\Gamma_\infty(A_i)-\Gamma_\infty(A)\big) W_{-r}\bigg)\\
        &\leq N\mathrm{Tr}(W_{-r}^\top\Gamma_\infty(A) W_{-r}) + N\Vert \Gamma_\infty(A_i)-\Gamma_\infty(A)\Vert_2\mathrm{Tr}(W_{-r} W_{-r}^\top),
        \end{align*}
        
        where we applied the trace inequality found in Theorem 1 of \cite{mori2002comments}, $\mathrm{Tr}(AB)\leq s_1(A)\mathrm{Tr}(B)$. 
        
        Hence,       
        \begin{align*}
        \mathrm{Tr}(W_{-r}^\top\mathbb{E}_{A_i}[X^\top X] W_{-r})&\leq Nr\underline{\lambda}_{r}(\Gamma_\infty(A))) + 16Nr\Vert A_i-A\Vert_2\Vert \Gamma_\infty(A)\Vert_2^3\\
        &\leq Nr\underline{\lambda}_{r}(\Gamma_\infty(A)) + 16Nr \\
        &\leq 17Nr \underline{\lambda}_{r}(\Gamma_\infty(A)),
    \end{align*}
    
    where the last inequality comes from $\underline{\lambda}_{r}(\Gamma_\infty(A))\geq 1$ since $\Gamma_{\infty}(A)\succeq I$.

    Finally, we obtain the following upper bound in the system identification case

    $$\EE_{A_i} [L_N(A_i,A)] \leq \frac{2C^2\varepsilon^2}{\sigma^2r} 17Nr \underline{\lambda}_{r}(\Gamma_\infty(A))\leq \frac{40NC^2\varepsilon^2\underline{\lambda}_{r}(\Gamma_\infty(A))}{\sigma^2}.$$
\end{itemize}

\textbf{Application of Lemma \ref{lemma:data_process}}: Consider

$$
    \forall i\in[ 2^{rd_y}]\quad \mathcal{E}_i=\{\Vert A_i-\hat A_n\Vert_\F \leq \varepsilon\}\quad \text{and}\quad \mathcal{E}=\{\Vert A-\hat A_n\Vert_\F \leq \varepsilon\}.
$$

Given $N$ samples, by stability of $\hat A_n$, one has for $i\in[ 2^{rd_y}],\ \mathbb{P}_{A_i}(\mathcal{E}_i)\geq 1-\delta$ and $ \mathbb{P}_{A}(\mathcal{E})\geq 1-\delta.$ 

Furthermore, these confusing models verify

$$
\forall i\in[2^{r d_y}], \quad \Vert A_i-A\Vert_\F=\frac{2C\varepsilon}{\sqrt{r}}\Vert Q_iW_{-r}^\top\Vert_\F 
= 2C\varepsilon 
\geq 2\varepsilon, $$

since both $Q_i, W_{-r}$ are semi-orthogonal and the Frobenius norm is unitarily invariant. Similarly, we have 

$$\forall i\neq j\in[2^{r d_y}], \quad 
\Vert A_i-A_j\Vert_\F
=\frac{2C\varepsilon}{\sqrt{r}} \Vert (Q_i-Q_j)W_{-r}^\top\Vert_\F 
=\frac{2C\varepsilon}{\sqrt{r}} \Vert Q_i-Q_j\Vert_\F 
\geq 2\varepsilon, $$

where the last inequality holds by Lemma \ref{lem:stiefel}. Since the $A_i$ are all distant from each other by at least $2\varepsilon $ then the $\mathcal{E}_i$ are all pairwise disjoint. Furthermore, since the $A_i$ are also distant from $A$ by at least $2\varepsilon $ then $\cup_{i=1}^{2^{r d_y}} \mathcal{E}_i \subset \mathcal{E}^c$. We can now apply Lemma \ref{lemma:data_process} and obtain 

\begin{itemize}
    \item For i.i.d multivariate regression: 
    \begin{align*}
    &\frac{1}{2}\log \left(\frac{2^{r d_y}}{4\delta}\right) \leq \frac{1}{2}\log\left(\frac{|\mathcal{P}_{r}^{d_y}|}{4\delta}\right) \leq  \sum_{i=1}^{2^{r d_y}}\frac{\EE_{A_i}[L_N(A_i,A)]}{2^{rd_y}} \\  
    &\implies \quad r d_y\log(2)+\log (\frac{1}{4\delta})\leq\frac{4NC^2\varepsilon^2 \underline{\lambda}_r(\Sigma)}{\sigma^2} \\
    &\implies \quad \log(2)(r d_y+\log (\frac{1}{\delta})-1)\leq\frac{4NC^2\varepsilon^2 \underline{\lambda}_r(\Sigma)}{\sigma^2} \\
    &\implies \frac{\log 2}{2}(rd_y+\log(\frac{1}{\delta}))\leq \frac{4NC^2\varepsilon^2 \underline{\lambda}_r(\Sigma)}{\sigma^2}.
    \end{align*} 
    
    This implies the following lower bound 

    $$ N\ge \frac{\sigma^2\log(2)}{8}\frac{\log (\frac{1}{\delta})+rd_y}{\varepsilon^2\underline{\lambda}_{r}(\Sigma)}\gtrsim \sigma^2\frac{\log (\frac{1}{\delta})+rd_y}{\varepsilon^2\underline{\lambda}_{r}(\Sigma)}.$$

    \item For dynamical systems:   
    \begin{align*}
        \frac{1}{2}\log \left(\frac{2^{r d_y}}{4\delta}\right) \leq \frac{1}{2}\log\left(\frac{|\mathcal{P}_{r}^{d_y}|}{4\delta}\right) &\leq  \sum_{i=1}^{2^{r d_y}}\frac{\EE_{A_i}[L_N(A_i,A)]}{2^{rd_y}}\\
        \implies \quad rd_y \log(2)+\log (\frac{1}{4\delta})&\leq\frac{80NC^2\varepsilon^2\underline{\lambda}_{r}(\Gamma_\infty(A))}{\sigma^2}.
    \end{align*}
    This implies the following lower bound

$$N\ge \frac{\sigma^2\log(2)}{160}\frac{\log (\frac{1}{\delta})+rd_y}{\varepsilon^2\underline{\lambda}_{r}(\Gamma_\infty(A))}\gtrsim \sigma^2\frac{\log (\frac{1}{\delta})+rd_y}{\varepsilon^2\underline{\lambda}_{r}(\Gamma_\infty(A))}.$$
\end{itemize}

\textbf{2) Lower bound that depends on $d_x$.} This case is relevant only for the multivariate regression since in this case, $d_x$ might differ from $d_y$. 
 
\textbf{Application of Lemma \ref{lem:log-lik}}: We use the confusing models from ${\cal C}_{2}(A,r,\varepsilon)$. Introduce for all $i\in[2^{r d_y}]$, 

$$
    A_i=A+\frac{2C\varepsilon}{\sqrt r} U_rR_i^\top , \qquad \text{where } R_i\in \mathcal{P}_{r}^{d_x}. 
$$

Then
\begin{align*}\EE_{A_i} [L_N(A_i,A)] &=  \frac{N}{2\sigma^2}\mathrm{Tr}((A-A_i)^\top(A-A_i)\Sigma) =\frac{4NC^2\varepsilon^2}{2\sigma^2r}\mathrm{Tr}(U_rR_i^\top\Sigma R_i U_r^\top) \\
    &=\frac{2NC^2\varepsilon^2}{\sigma^2r}\mathrm{Tr}(R_i^\top\Sigma R_i) \leq \frac{2NC^2\varepsilon^2 \bar{\lambda}_r(\Sigma)}{\sigma^2} \quad \text{by Lemma \ref{lem:min_max}.}  
    \end{align*}

\textbf{Application of Lemma \ref{lemma:data_process}}: Using the same arguments as to derive the lower bound that depends on $d_y$, we obtain
\begin{align*}
    \frac{1}{2}\log \left(\frac{2^{r d_x}}{4\delta}\right) \leq \frac{1}{2}\log\left(\frac{|\mathcal{P}_{r}^{d_x}|}{4\delta}\right) &\leq  \sum_{i=1}^{2^{r d_x}}\frac{\EE_{A_i}[L_N(A_i,A)]}{2^{rd_x}} \\    \implies \quad r d_x\log(2)+\log (\frac{1}{4\delta})&\leq\frac{4NC^2\varepsilon^2 \bar{\lambda}_r(\Sigma)}{\sigma^2}. 
    \end{align*} 
    
This implies the following lower bound 

    $$ N \gtrsim \sigma^2\frac{\log (\frac{1}{\delta})+rd_x}{\varepsilon^2\bar{\lambda}_{r}(\Sigma)}.$$

Combining both lower bounds on the sample complexity, we obtain the stated result.
\end{proof}

\subsection{Proof of Theorem \ref{thm:reduced}} \label{app:sec4:thm2}

\begin{proof}
Let $\varepsilon>0$ and $\delta\in(0,1).$ We suppose that $\EE(\hat \Sigma)\succ 0$. Let $A$ such that $ \mathrm{rank}(A)\leq r \leq \frac{1}{2}\underline{d}$.  Let $\hat A_n$ be an $(\varepsilon,\delta)$-stable algorithm in $A$ and let $N$ be its sample complexity. By stability assumption, $N$ verifies $\Vert A-\Pi_{k^\star_{A,N}}(A) \Vert_\F\leq \varepsilon$. Hence, letting $\varepsilon'=2\varepsilon-\Vert A-\Pi_{k^\star_{A,N}}(A) \Vert_\F \geq\varepsilon $, we also have by stability 

$$\forall A'\in \mathcal{D}(A,N,\varepsilon), \quad \mathbb{P}_{ A'}(\Vert \hat A_n-A'\Vert _\F\leq\varepsilon')\geq 1-\delta.$$ 

The remainder of the proof is identical to that of Theorem \ref{thm:full}. We only provide the proof for the multivariate regression case (simply replace $\Sigma$ by $\Gamma_\infty(A)$ for dynamical systems, and adjust the universal constant).

\textbf{1) Lower bounds that depends on $d_y$}: Consider the confusing models from ${\cal C}_{1}(A,k^\star_{A,N},\varepsilon)$:

$$\forall i\in[2^{k^\star_{A,N}d_y}], \quad A_i=A+\frac{2C\varepsilon'}{\sqrt {k^\star_{A,N}}} Q_iW_{-k^\star_{A,N}}^\top, \quad \text{ where } Q_i\in \mathcal{P}_{k^\star_{A,N}}^{d_y}.  $$

We apply Lemmas \ref{lem:log-lik} and \ref{lemma:data_process} as in the proof of Theorem \ref{thm:full} (replace $r$ by $k^\star_{A,N}$, and $\varepsilon$ by $\varepsilon'$), and derive following lower bound:

    $$\varepsilon'^2\ge \frac{\sigma^2\log(2)}{8}  \frac{\log (\frac{1}{\delta})+k^\star_{A,N}d_y}{N\underline{\lambda}_{k^\star_{A,N}}(\Sigma)}.$$


    
    Since $\varepsilon\ge\Vert A-\Pi_{k^\star_{A,N}}(A) \Vert_\F$, we have:
    \begin{align*} \varepsilon'^2&=4\varepsilon^2+\Vert A-\Pi_{k^\star_{A,N}}(A) \Vert_\F^2-4\varepsilon\Vert A-\Pi_{k^\star_{A,N}}(A) \Vert_\F\\
&\leq 4\varepsilon^2 - 3\Vert A-\Pi_{k^\star_{A,N}}(A) \Vert_\F^2.
\end{align*} 

Combining the two previous inequalities yields: 
\begin{align*}
&4\varepsilon^2\geq \frac{\sigma^2\log(2)}{8}\frac{\log (\frac{1}{\delta})+k^\star_{A,N}d_y}{N\underline{\lambda}_{k^\star_{A,N}}(\Sigma)}+3\Vert A-\Pi_{k^\star_{A,N}}(A) \Vert_\F^2\\
&\implies \frac{32}{\log(2)}\varepsilon^2\geq \gamma_A^\delta(N),
\end{align*}
This finally implies that: 

$$N\geq\min\{n:\gamma_A^\delta(n)\leq \frac{32}{\log(2)}\varepsilon^2\}. $$

\textbf{2) Lower bound that depends on $d_x$}: The proof follows along the same lines as those in the proof of Theorem \ref{thm:full}. 

\end{proof}


\subsection{Extremal partial trace}

\begin{lemma}

For any $k\leq d$ and two matrices $Q\in \mathrm{St}_k^d(\mathbb{R}),\Sigma\in\mathbb{R}^{d\times d}$ symmetric positive definite, we have 
\begin{equation}
\sum_{i=d-k+1}^d\lambda_i(\Sigma)\leq \mathrm{Tr}(Q^\top\Sigma Q)\leq \sum_{i=1}^k\lambda_i(\Sigma).
\end{equation}
\label{lem:min_max}
\end{lemma}

\begin{proof}
    The proof of this classical result, sometimes referred to as \textit{extremal partial trace}, can be found in 20.A.2 of \cite{marshall1979inequalities}.
\end{proof}

\newpage
\section{Proofs of results presented in Section \ref{sec:denoising}}\label{app:sec5}

In this section, we first consider an estimator $\bar{A}$ of $A$ and derive a tight upper bound on $\Vert \Pi_k(\bar{A}) - A\Vert_\F$ for arbitrary values of $k$ (Lemma \ref{lem:trade_off_k}), highlighting the trade-off between estimation error and approximation error. We then address the problem of selecting the optimal value of $k$ to minimize this upper bound. While solving this optimization problem exactly would require knowledge of $A$, which is unknown by definition, we propose a universal thresholding rule that bypasses this limitation (Theorem \ref{thm:opt_best_k}). Finally in \textsection \ref{subsec:nn}, we explain how to apply the thresholding procedure to improve existing estimators. We illustrate the procedure for the estimator obtained via regression with nuclear norm regularization.

\subsection{Proof of Lemma \ref{lem:trade_off_k}}
\begin{proof} 

Let $k\in[\underline{d}]$, we have by the Eckart-Young's theorem:

$$
\Vert \Pi_k(A+ Z) -A-Z\Vert_\F^2 \leq  \Vert \Pi_k(A) - A - Z \Vert_\F^2.
$$ 

By decomposing the squares, we obtain: 

$$
    \Vert \Pi_k(A+Z) -A \Vert^2 \leq \Vert A-\Pi_k(A)\Vert_\F^2  +2 \langle \Pi_k(A+Z) -\Pi_k(A), Z \rangle.
$$

For the l.h.s, we have, by reverse triangle inequality,

\begin{align*}
    \Vert \Pi_k(A+Z) -A \Vert^2 
    &\geq \left| \Vert \Pi_k(A+Z) -\Pi_k(A)\Vert_\F -  \Vert A -\Pi_k(A)\Vert_\F \right|^2\\
    &=\Vert \Pi_k(A+Z) -\Pi_k(A)\Vert_\F^2 + \Vert A -\Pi_k(A)\Vert_\F^2 - 2\Vert \Pi_k(A+Z) -\Pi_k(A)\Vert_\F\Vert A -\Pi_k(A)\Vert_\F. 
\end{align*}

Hence, we get

$$\frac{1}{2}\Vert \Pi_k(A+Z) -\Pi_k(A)\Vert_\F^2\leq  \langle \Pi_k(A+Z) -\Pi_k(A), Z \rangle + \Vert \Pi_k(A+Z) -\Pi_k(A)\Vert_\F\Vert A -\Pi_k(A)\Vert_\F.$$

Furthermore, since $\Pi_k(A+Z) -\Pi_k(A)$ is of rank at most $2k$, we have by Lemma 2 of \cite{xiang2012optimal}

$$    \langle \Pi_k(A+Z) -\Pi_k(A), Z \rangle 
\leq \Vert \Pi_k(A+Z) -\Pi_k(A)\Vert_\F \Vert \Pi_{2k}(Z)\Vert_\F.
$$

This implies: 

$$        \frac{1}{2}\Vert \Pi_k(A+Z) -\Pi_k(A)\Vert_\F \leq \Vert \Pi_{2k}(Z)\Vert_\F+ \Vert A -\Pi_k(A)\Vert_\F.
$$

We conclude by noting that:

\begin{align*}
\Vert \Pi_k(\bar A)-A\Vert_\F &\leq \Vert \Pi_k(A+Z) -\Pi_k(A)\Vert_\F + \Vert A -\Pi_k(A)\Vert_\F \\
&\leq 2\Vert \Pi_{2k}(Z)\Vert_\F + 3\Vert A -\Pi_k(A)\Vert_\F.
\end{align*}
\end{proof}

\subsection{Proof of Theorem \ref{thm:opt_best_k}}
\begin{proof} 

We first prove an auxiliary useful inequality.

Let $\alpha\geq 0$ and $K=\max\{i:s_i(\bar A)\geq (2+\alpha)\Vert Z\Vert_2\}$. Further define $k^{\star}=\max\{i:s_i(A)\geq (3+\alpha)\Vert Z\Vert_2\}$. By Weyl's inequality, we have:

\begin{itemize}
    \item $s_{K+1}(A)\leq s_{K+1}(\bar A) + \Vert Z\Vert_2\leq (3+\alpha)\Vert Z \Vert_2$, which implies that $k^{\star}\leq K$. 
    \item $\forall i\leq K, \ s_i(A)\geq s_i(\bar A)-\Vert Z\Vert_2\geq (1+\alpha)\Vert Z\Vert_2 $.
\end{itemize}  
By Lemma \ref{lem:trade_off_k}:
\begin{align}
    \frac{\Vert \Pi_{K}(\bar A)-A\Vert_\F^2}{18}
    &\leq K\Vert Z\Vert_2^2+ \sum_{i>K}s_i^2(A) = k^{\star}\Vert Z\Vert_2^2+\sum_{i>k^{\star}}s_i^2(A) +\psi\nonumber\\
    &\leq k^{\star}\Vert Z\Vert_2^2+\sum_{i>k^{\star}}s_i^2(A) \nonumber\\
    &\leq (3+\alpha)^2k^{\star}\Vert Z\Vert_2^2+\sum_{i>k^{\star}}s_i^2(A) \nonumber\\
    &= \min_{k\in[r]} \left\{(3+\alpha)^2k\Vert Z\Vert_2^2+ \sum_{i=k+1}^rs_i^2(A)\right\},\label{eq:ffr} 
\end{align}
where $\psi=\sum_{i=k^{\star}+1}^{K}\left(\Vert Z\Vert_2^2-s_i^2(A)\right)<0$. Let $\xi\geq 2\Vert Z\Vert_2$ and $\alpha=\frac{\xi}{\Vert Z\Vert_2}-2$. Then we have: $\xi =(2+\alpha)\Vert Z\Vert_2$. The Theorem is obtained by combining (\ref{eq:ffr}) and the observation that $(3+\alpha)\Vert Z\Vert_2=(1+\frac{\xi}{\Vert Z\Vert_2})\Vert Z\Vert_2\leq 2\xi$.
\end{proof}

\subsection{Thresholded nuclear norm estimator}\label{subsec:nn}

We demonstrate how Theorem \ref{thm:opt_best_k} can be applied to devise a rank-adaptive algorithm, starting from an estimator $\bar{A}$ obtained via regression with nuclear norm regularization. We restrict our attention to the case of multivariate regression with Gaussian noise. The estimator $\bar{A}$ is defined in \cite{negahban2011estimation} by: 

$$
\bar A = \arg\min_{B} \{\Vert Y-XB\Vert_\F ^2+\mu \Vert B\Vert_1\}.
$$
The performance of this estimator is well understood (see Corollary 10.14 in \cite{wainwright2019high}): if $\bar{A}$ is constructed using $n$ samples and with $\mu=10\sigma\sqrt{\lambda_\mathrm{max}(\hat \Sigma)}(\sqrt{\frac{d_x+d_y}{n}}+\sqrt{\frac{\log (\frac{1}{2\delta})}{2n}})$, then the error $Z=\bar{A}-A$ satisfies, with probability at least $1-\delta$:

$$2\Vert Z\Vert_2\leq \xi: = \frac{20\sigma\sqrt{\lambda_\mathrm{max}(\hat \Sigma)}(\sqrt{\frac{d_x+d_y}{n}}+\sqrt{\frac{\log (\frac{1}{2\delta})}{2n}})}{\lambda_\mathrm{min}(\hat \Sigma)}.$$

Combining this concentration result to the result of Theorem \ref{thm:opt_best_k} yields: 

\begin{corollary}\label{coro:nuclear}
The {\it thresholded nuclear norm regularized estimator} $\hat{A}_n:=\bar{A}(\xi)$ satisfies, with probability at least $1-\delta:
\Vert \hat{A}_n-A\Vert_\F^2\lesssim \min_{k\in[r]}\bigg[ k\sigma^2\kappa(\hat\Sigma)\frac{d_x+d_y+\log(\frac{1}{\delta})}{n\lambda_\mathrm{min}(\hat\Sigma)}+ \sum_{i>k}s_i^2(A)\bigg].$
\end{corollary}

\begin{proof}
    We plug in the value of $\xi=\frac{20\sigma\sqrt{\lambda_\mathrm{max}(\hat \Sigma)}(\sqrt{\frac{d_x+d_y}{n}}+\sqrt{\frac{\log (\frac{1}{2\delta})}{2n}})}{\lambda_\mathrm{min}(\hat \Sigma)}$ inside the upper bound given by Theorem \ref{thm:opt_best_k}. 
\end{proof}

In \cite{bunea2011optimal}, the authors derive, in their Theorem 12, an error upper bound for the plain nuclear norm regularized estimator (without thresholding). Their upper bound is similar to the one presented in the above corollary, but is $\kappa(\hat\Sigma)$ times larger. Hence, incorporating the singular value thresholding procedure results in an algorithm with stronger performance guarantees.

\newpage
\section{Proofs of results presented in Section \ref{sec:lse}}\label{app:sec6}

In this section, we provide the proofs of all results presented in Section \ref{sec:lse}. These results concern the performance of our algorithms \RLSE\ and \TLSE. We first observe that all stated inequalities in Section \ref{sec:lse} are trivially true when $\hat{\Sigma}$ is singular since then each upper bound becomes infinity. Therefore, we will assume from now on that $\hat{\Sigma}$ is non-singular. In that case, the LSE is expressed as

$$\hat{A}_n = Y^\top X(X^\top X)^{-1}=A+E^\top X(X^\top X)^{-1}:=A+Z,$$

allowing us to apply denoising results from Section \ref{sec:denoising}. Since the norms considered are invariant by transposition, we will consider $Z:=(X^\top X)^{-1}X^\top E$ for notational convenience.

\subsection{Multivariate regression}

\subsubsection{Proof of Lemma \ref{lem:pi_k_lse} and Equation (\ref{eq:thresh_LR})}

\begin{proof} 

Let $\delta\in(0,1)$ and $k\in[\underline{d}]$. We have:
\begin{align*}
    Z &:= (X^\top X)^{-1} X^\top E = (X^\top X)^{-1} (X^\top X)(X^\top X)^{-1} X^\top E.
\end{align*}

Furthermore, 

\begin{itemize}
    \item Lemma 3 of \cite{bunea2011optimal} states: with probability at least $1-\delta$,
    $$\Vert X(X^\top X)^{-1} X^\top E\Vert_2 \leq \sigma\left(\sqrt{d_x}+\sqrt{d_y}+\sqrt{\log\left(\frac{1}{\delta}\right)}\right). $$
    \item For all matrices $P,Q$, we have 
    $$\Vert \Pi_k(PQ)\Vert_\F^2=\sum_{i=1}^ks_i^2(PQ)\leq \sum_{i=1}^ks_i^2(P)s_1^2(Q)=\Vert \Pi_k(P)\Vert_\F^2\Vert Q\Vert_2^2.$$

\end{itemize}

Hence, we can write
\begin{align*}
        \Vert \Pi_k(Z)\Vert_\F^2&\leq\Vert \Pi_k\left((X^\top X)^{-1} X^\top\right)\Vert_\F^2\Vert X(X^\top X)^{-1} X^\top E\Vert_2^2\\
        &\leq\left(\sum_{i=d_x-k+1}^{d_x}\frac{1}{n\lambda_i(\hat{\Sigma})}\right) \Vert X(X^\top X)^{-1} X^\top E\Vert_2^2\\
        &\leq \frac{\sum_{i=d_x-k+1}^{d_x}\frac{1}{\lambda_i(\hat \Sigma)}}{k}\frac{ k\sigma^2(\sqrt{d_x}+\sqrt{d_y}+\sqrt{\log(\frac{1}{\delta})})^2}{n},
    \end{align*}

which concludes the proof of Lemma \ref{lem:pi_k_lse}.
    
    The proof of Equation (\ref{eq:thresh_LR}) is similar and uses the sub-multiplicativity of the operator norm to obtain  
        \begin{align*}
        \Vert Z\Vert_2&\leq \Vert (X^\top X)^{-1} X^\top\Vert_2 \Vert X(X^\top X)^{-1} X^\top E\Vert_2\\
        &\leq \frac{1}{\sqrt{n\lambda_\mathrm{min}(\hat\Sigma)}}  \Vert X(X^\top X)^{-1} X^\top E\Vert_2\\
&\leq\frac{\sigma\left(\sqrt{d_x}+\sqrt{d_y}+\sqrt{\log(\frac{1}{\delta})}\right)}{\sqrt{n\lambda_\mathrm{min}(\hat\Sigma)}}.
    \end{align*}
\end{proof}

\subsubsection{General error bounds for any $\hat \Sigma$}
Before stating the proofs of our main Theorems \ref{thm:rlse-gauss} and \ref{thm:tlse-gauss} presented in our main body, we start by stating general concentration-free result, valid for any $\hat{\Sigma}$. 

\begin{theorem}
    Assume that $\mathrm{rank}(A)\leq r$, the \RLSE~satisfies, with probability at least $1-\delta$:
    
    $$ \Vert \hat{A}_n-A\Vert_\F^2\leq 6\sqrt 2 r\sigma^2 \frac{\bar d+\log(\frac{1}{\delta})} {n\underline{\lambda}^H_r(\hat \Sigma)}.$$
\label{thm:rlse-gen}
\end{theorem}

\begin{proof}
We first apply Lemma \ref{lem:pi_k_lse} with $k=r$ which gives with probability at least $1-\delta$
$$\Vert \Pi_r(Z)\Vert_\F^2\leq \frac{ r\sigma^2(\sqrt{d_x}+\sqrt{d_y}+\sqrt{\log(\frac{1}{\delta})})^2}{n\underline{\lambda}^H_r(\hat\Sigma)}.$$

We plug this inequality in Lemma \ref{lem:trade_off_k}, also applied to $k=r$, which gives 

\begin{align*}
\Vert \hat{A}_n - A\Vert_\F &=\Vert \Pi_r(\bar A)-A\Vert_\F \leq 2\sqrt{2}\Vert \Pi_{r}(Z)\Vert_\F \\
&\leq 2\sqrt{2}\frac{ r\sigma^2(\sqrt{d_x}+\sqrt{d_y}+\sqrt{\log(\frac{1}{\delta})})^2}{n\underline{\lambda}^H_r(\hat\Sigma)}\leq C_\RLSE \frac{ r\sigma^2(d_x+d_y+\log(\frac{1}{\delta}))}{n\underline{\lambda}^H_r(\hat\Sigma)},
\end{align*}

where $C_\RLSE=6\sqrt{2}$.

\end{proof}

In the adaptive case, one has the following 

\begin{theorem}
   The \TLSE~ $\hat{A}_n:=\bar{A}(\xi_{\textup{MR}})$ satisfies, with probability at least $1-\delta$:
   
    $$
\Vert \hat{A}_n-A\Vert_\F^2\leq 864\min_{k\in[r] }\left( k\sigma^2\frac{\bar d+\log(\frac{1}{\delta})}{n\lambda_\mathrm{min}(\hat \Sigma)}+ \sum_{i>k}s_i^2(A)\right).
$$
\label{thm:tlse-gen}
\end{theorem}

\begin{proof}

We proved above that $\xi_{\textup{MR}} :=  \frac{2\sigma\left(\sqrt{d_x}+\sqrt{d_y}+\sqrt{\log(\frac{1}{\delta})}\right)}{\sqrt{n\lambda_\mathrm{min}(\hat\Sigma)}}\geq 2\Vert Z\Vert_2$ with probability at least $1-\delta$. We can plug this value inside the upper bound given by Theorem \ref{thm:opt_best_k} to obtain 

\begin{align*}
    \Vert \hat{A}_n - A\Vert_\F^2&=\Vert \bar A(\xi_{\textup{MR}}) - A\Vert_\F\\ 
    &\leq 18\min_{k\in[r]}\left( 4 k\xi_{\textup{MR}}^2+ \sum_{i>k}s_i^2(A)\right) \\
    &\leq 18\min_{k\in[r]}\left( 4k\frac{4\sigma^2\left(\sqrt{d_x}+\sqrt{d_y}+\sqrt{\log(\frac{1}{\delta})}\right)}{n\lambda_\mathrm{min}(\hat\Sigma)}^2+ \sum_{i>k}s_i^2(A)\right) \\
    &\leq C_\TLSE\min_{k\in[r]}\left( \frac{k\sigma^2\left(d_x+d_y+\log(\frac{1}{\delta})\right)}{n\lambda_\mathrm{min}(\hat\Sigma)}+ \sum_{i>k}s_i^2(A)\right),
\end{align*}

where $C_\TLSE=18\times16\times3=864$.
\end{proof}

To obtain matching lower and upper bounds, we need to further apply concentration results uniformly on the spectrum of $\hat \Sigma$. Such results have recently been studied in \cite{barzilai2024simple} and applied to Gaussian inputs. We collect some of their findings in Appendix \ref{app:stability} and leverage them to prove Theorems \ref{thm:rlse-gauss} and \ref{thm:tlse-gauss}. 

\subsubsection{Performance of \RLSE: proof of Theorem \ref{thm:rlse-gauss}}

\begin{proof}
    Let $n\geq 288(d_x+\log(\frac{2}{\delta}))$. We first apply Theorem \ref{thm:rlse-gen} which gives with probability at least $1-\delta/2$:
    $$ \Vert \hat{A}_n-A\Vert_\F^2\leq 6\sqrt 2 r\sigma^2 \frac{\bar d+\log(\frac{2}{\delta})} {n\underline{\lambda}^H_r(\hat \Sigma)}.$$
    
    By Lemma \ref{lem:cov-concentration-1}, we can replace $\hat \Sigma$ by $\Sigma$ up to a universal constant with probability at least $1-\delta/2$.

    Combining both results concludes the proof of Theorem \ref{thm:rlse-gauss}.
    
\end{proof} 

\subsubsection{Performance of \TLSE: proof of Theorem \ref{thm:tlse-gauss}}
\begin{proof}
Let $n\geq 288(d_x+\log(\frac{2}{\delta}))$. We first apply Theorem \ref{thm:tlse-gen} which gives with probability at least $1-\delta/2$:
    $$ \Vert \hat{A}_n-A\Vert_\F^2\leq 864\min_{k\in[r] }\left( k\sigma^2\frac{\bar d+\log(\frac{2}{\delta})}{n\lambda_\mathrm{min}(\hat \Sigma)}+ \sum_{i>k}s_i^2(A)\right).$$
    
    By Lemma \ref{lem:cov-concentration-1}, we can replace $\hat \Sigma$ by $\Sigma$ up to a universal constant with probability at least $1-\delta/2$.

    Combining both results concludes the proof of Theorem \ref{thm:tlse-gauss}.

\end{proof}
\subsection{Linear system identification}

\subsubsection{Proof of Lemma \ref{lem:pi_k_lse_lds} and  Equation (\ref{eq:thresh_LDS})}
\begin{proof} 

Let $n$ verifying inequality (\ref{eq:lti_condition}). 

We first prove (i).
Let $\delta\in (0,1)$ and $k\in[\underline{d}]$.

We use a slightly different decomposition of the LSE error compared to the multivariate regression case. This decomposition will allow us to apply recent concentration results.

$$
    Z = (X^\top X)^{-1} X^\top E = (X^\top X)^{-\frac{1}{2}} (X^\top X)^{-\frac{1}{2}} X^\top E.
$$

In the proof of Theorem 3 of \cite{jedra2022finite}, it is shown that

    \begin{equation}
        \mathbb{P}\left[\Vert (X^\top X)^{-\frac{1}{2}} X^\top E\Vert_2\geq \sigma\sqrt{d_x+\log\left(\frac{1}{\delta}\right)}\right]\leq \delta,
        \label{eq:err-lti}
    \end{equation} 

    when the following condition is verified
    
  $$\lambda_\mathrm{min}\left(\sum_{i=0}^{n-2}\Gamma_i(A)\right)\ge c_0\sigma^4\left(\sum_{i\geq 0}\Vert A^i\Vert_2\right)^2\left(d_x+\log\left(\frac{1}{\delta}\right)\right).$$

  for some universal constant $c_0>0$.

  Applying items (i) and (iii) of Lemma \ref{lem:properties-Gamma}, we can replace $\lambda_\mathrm{min}\left(\sum_{i=0}^{n-2}\Gamma_i(A)\right)$ by $n\lambda_\mathrm{min}(\Gamma_\infty(A))$ and $\left(\sum_{i\geq 0}\Vert A^i\Vert_2\right)^2$ by $\Vert \Gamma_\infty(A)\Vert_2^3$ up to some positive universal constants which we incorporate into $c_0$. Therefore, verifying inequality (\ref{eq:lti_condition}) is sufficient to obtain Equation \ref{eq:err-lti}.

    In that case, using similar arguments as in the proof of Lemma \ref{lem:pi_k_lse}, we have with probability at least $1-\delta$,
    
    \begin{align*}
        \Vert \Pi_k(Z)\Vert_\F^2&\leq\Vert \Pi_k\left((X^\top X)^{-\frac{1}{2}}\right)\Vert_\F^2\Vert (X^\top X)^{-\frac{1}{2}} X^\top E\Vert_2^2\\        
        &\leq \frac{\sum_{i=d_x-k+1}^{d_x}\frac{1}{\lambda_i(\hat \Sigma)}}{k}\frac{ k\sigma^2(d_x+\log(\frac{1}{\delta}))^2}{n},
    \end{align*}

We now prove (ii). Since $\Gamma_\infty(A) \succ 1$ then $\frac{\Vert \Gamma_\infty(A)\Vert_2^3}{\lambda_{\min}(\Gamma_\infty(A))}\geq \Vert \Gamma_\infty(A) \Vert_2$. Therefore verifying inequality (\ref{eq:lti_condition}) ensures that $n\geq 2\log(2)\Vert \Gamma_\infty(A)\Vert_2$. By (iii) of Theorem \ref{thm:cov-concentration-2}, we obtain the desired result. This concludes the proof of Lemma \ref{lem:pi_k_lse_lds}.

The proof of Equation (\ref{eq:thresh_LDS}) uses the sub-multiplicativity of the operator norm to obtain 
\begin{align*}
    \Vert Z\Vert_2 
    &\leq  \Vert(X^\top X)^{-\frac{1}{2}}\Vert_2 \Vert (X^\top X)^{-\frac{1}{2}} X^\top E\Vert_2 \leq \sigma\sqrt{\frac{ d_x+\log(\frac{1}{\delta})}{n\lambda_\mathrm{min}(\hat \Sigma)}}.
\end{align*}

    \end{proof}

\subsubsection{Performance of \RLSE: proof of Theorem \ref{thm:rlse-lti} }
\begin{proof}

We first apply both (i) and (ii) of Lemma \ref{lem:pi_k_lse_lds} with $k=r$, which gives with probability at least $1-\delta$,
$$\Vert \Pi_r(Z)\Vert_\F^2\lesssim \frac{ r\sigma^2(\sqrt{d_x}+\sqrt{d_y}+\sqrt{\log(\frac{1}{\delta})})^2}{n\underline{\lambda}^H_r(\Gamma_\infty(A))}.$$

We plug this inequality in Lemma \ref{lem:trade_off_k}, also applied to $k=r$, which gives $\Vert \Pi_r(\bar A)-A\Vert_\F \leq 2\sqrt{2}\Vert \Pi_{r}(Z)\Vert_\F $ to obtain the desired result.   
\end{proof}

\subsubsection{Performance of \TLSE: proof of Theorem \ref{thm:tlse-lti}}

\begin{proof}
We proved above that $\xi_{\textup{SysID}} :=  2\sigma\sqrt{\frac{ d_x+\log(\frac{1}{\delta})}{n\lambda_\mathrm{min}(\hat \Sigma)}}\geq 2\Vert Z\Vert_2$ with probability at least $1-\delta$. We can plug this value inside the upper bound given by Theorem \ref{thm:opt_best_k} to obtain 

$$
\Vert \hat{A}_n-A\Vert_\F^2\lesssim \min_{k \in [r]}\left( \frac{k\sigma^2(d_x+\log(\frac{1}{\delta}))}{n\lambda_\mathrm{min} (\hat \Sigma)}+ \sum_{i>k}s_i^2(A)\right).
$$

We then use (ii) of Lemma \ref{lem:pi_k_lse_lds} to replace $\hat \Sigma $ by $\Gamma_\infty(A)$ up to a positive universal constant and obtain the desired result.
\end{proof}

\newpage
\section{Stability of \RLSE\ and \TLSE\ 
}\label{app:stability}

{
In this section, we present results showing that \RLSE~and \TLSE~are stable in the sense of the definitions presented in \textsection\ref{sec:lb}. The results follow from Theorems \ref{thm:rlse-gauss}   and Theorem \ref{thm:tlse-gauss} but require the intermediate step of establishing concentration bounds on the spectrum empirical covariance matrix $\hat{\Sigma}$ in both the multivariate regression and system identification settings.  In Lemma \ref{lem:cov-concentration-1} and Theorem \ref{thm:cov-concentration-2}, we present concentration on the spectrum of such matrices. Equipped with these concentration bounds we can revisit our guarantees for \RLSE~and \TLSE~ to provide upper bounds that only exhibit a dependence on the true covariance matrix $\Sigma$ in the case multivariate regression and on $\Gamma_\infty(A)$ in the case of linear system identification. 
}


\subsection{Concentration of the empirical covariance matrix}
To obtain sample complexity upper bounds for   \RLSE~and \TLSE~, we need to derive concentration bounds on the entire spectrum of $\hat{\Sigma}$. The derivation of such concentrations depends a lot on the collection of random variables $x_1, \dots, x_n$.

\paragraph{Multivariate regression.} When $x_1, \dots, x_n$ are i.i.d. gaussian random variables with zero mean and covariance $\Sigma$, we have the following result.


\begin{lemma}[Theorem 8 of \cite{barzilai2024simple}]\label{lem:cov-concentration-1}
    Let $X\in\mathbb{R}^{n\times d_x}$, with $d_x\leq n$, be a matrix whose rows $x_i\sim \mathcal{N}(0,\Sigma)$ are i.i.d. Then for any $\delta>0$, it holds with probability at least $1-\delta$,
    
    $$\forall i\in[d_x], \quad |\lambda_i(\hat{\Sigma})-\lambda_i(\Sigma)|\leq 2\varepsilon+\varepsilon^2, \quad \text{where} \quad \varepsilon=\sqrt{\frac{d_x}{n}}+\sqrt{\frac{2\log(\frac{1}{\delta})}{n}}.$$

    
    In particular, for any $\delta > 0$, provided $n\geq 288(d_x+\log(\frac{1}{\delta}))$, it holds with probability at least $1-\delta$,
    
    $$\forall i\in[d], \quad \frac{1}{2}\lambda_i(\Sigma)\leq\lambda_i(\hat{\Sigma}).$$
\end{lemma}

\paragraph{Linear system identification.} In this setting, the random variables $x_1, \dots, x_n$ are dependent, which makes the analysis of $\hat{\Sigma}$ more challenging, especially if we need a concentration result on the entire spectrum. We present the following result. 
\begin{theorem}\label{thm:cov-concentration-2} 
Let $\delta > 0$. Assume that for all $i\ge 0$, $x_{i+1} = A x_i + \eta_i$, where $x_0 = 0$, and $(\eta_i)_{i\ge 0}$ are i.i.d. zero-mean, $\sigma^2$-subgaussian random variables. Recall that $\hat{\Sigma} = \frac{1}{n} X^\top X$ and $\Sigma = \frac{1}{n}\sum_{i=0}^{n-1}\Gamma_i(A)$. Then, with probability $1-\delta$, it holds 

\begin{align*}
    \forall i \in [d_x], \qquad  \frac{1}{2}\lambda_i(\Sigma) \le \lambda_i(\hat{\Sigma}) \le \frac{3}{2}\lambda_i(\Sigma) 
\end{align*}

provided that 
$
n \ge \frac{ c_1 \sigma^4 \Vert \Gamma_\infty(A) \Vert_2^3 }{\lambda_{\min}(\Sigma) } \left( \log\left(\frac{1}{\delta}\right) + c_2 d_x\right)
$
where $c_1, c_2$ are positive universal constants.  

In particular, with probability $1-\delta$, it holds that 

\begin{align*}
    \forall i \in [d_x], \qquad  \frac{1}{8}\lambda_i(\Gamma_{\infty}(A)) \le \lambda_i(\hat{\Sigma}) \le \frac{3}{2}\lambda_i(\Gamma_\infty(A)) 
\end{align*}

provided that $
n \ge \frac{ c_3 \sigma^4 \Vert \Gamma_\infty(A) \Vert_2^3 }{\lambda_{\min}(\Gamma_\infty(A))} \left( \log\left(\frac{1}{\delta}\right) + c_4 d_x\right)
$ where $c_3, c_4$ are positive universal constants.  
\end{theorem}

Before proving this theorem, we need to review some key intermediate lemmas. First, we provide the following result which lists certain properties about the controllability Gramians.

\bigskip
\begin{lemma}[Lemma 8 in \cite{jedra2022finite}]\label{lem:properties-Gamma} Assume that $A$ is stable, meaning its spectral radius satisfies  $\rho(A)<1$. Then, the following properties hold
\begin{itemize}
    \item [(i)] 
    $\sum_{t=0}^\infty \left\Vert A^t \right\Vert_2 \le (1+\sqrt{2}) \Vert \Gamma_\infty(A)\Vert_2^{3/2}$
    \item [(ii)] for all $i\ge 0$,  $\left(1- \exp(-\frac{i+1}{ \Vert \Gamma_\infty(A) \Vert_2  - 1})\right)\Gamma_\infty(A) \preceq \Gamma_i(A) \preceq \Gamma_\infty(A)$ \\
    \item [(iii)] for $n \ge 2\log(2) \Vert \Gamma_\infty(A)\Vert_2 $, 
    we have $ \frac{1}{4} \Gamma_\infty(A) \preceq \frac{1}{n}\sum_{i=0}^{n-1} \Gamma_i(A) \preceq  \Gamma_\infty(A)$.
\end{itemize}
\end{lemma}
\begin{proof}[Proof of Lemma \ref{lem:properties-Gamma}]
The results corresponding of \emph{(i)} and \emph{(ii)} are borrowed from Lemma 8 in \cite{jedra2022finite}. We will therefore omit their proofs here. To prove \emph{(iii)}, note that $\Gamma_\infty(A) \succeq \frac{1}{n} \sum_{i=0}^{n-1} \Gamma_i(A)$ is immediate from \emph{(ii)}. Let us then prove the lower bound which also follows from  \emph{(ii)}. Note that 

\begin{align*}
    \Sigma & = \frac{1}{n}\sum_{i=0}^{n-1} \Gamma_{i}(A) \succ \frac{1}{n} \sum_{i=\lfloor n/2 \rfloor}^{n-1} \Gamma_{i}(A) \\
   & \succeq \frac{1}{n}\left(\sum_{i=\lfloor n/2 \rfloor}^{n-1} 
  (1 - e^{-\frac{i+1}{\Vert \Gamma_\infty(A) \Vert_2 - 1} })\Gamma_\infty(A)\right) \\
  & \succeq \frac{1}{n}\left( 1 - e^{ - \frac{\lfloor n/2\rfloor+1 }{\Vert \Gamma_\infty(A) \Vert_2 - 1 }}\right)   (n - \lfloor n/2\rfloor)\Gamma_\infty(A) \\
  & \succeq \frac{1}{2}\left( 1 - e^{ - \frac{\lfloor n/2\rfloor+1 }{\Vert \Gamma_\infty(A) \Vert_2 - 1 }}\right) \Gamma_\infty(A).
\end{align*}
We also have 
\begin{align*}
    n \ge 2 \log(2)\Vert \Gamma_\infty(A)  \Vert_2 \implies  \left( 1 - e^{ - \frac{\lfloor n/2\rfloor+1 }{\Vert \Gamma_\infty(A) \Vert_2 - 1 }}\right) \ge \frac{1}{2}.
\end{align*}
Thus, whenever the above condition holds it follows that  
\begin{align*}
    \Sigma \succeq \frac{1}{4} \Gamma_\infty(A).
\end{align*}
\end{proof}

Next, we present a concentration result which is central in the proof of Theorem \ref{thm:cov-concentration-2}.
\begin{lemma}[Lemma 4 in \cite{jedra2022finite}]\label{lem:cov-concentration-2} Let $\varepsilon, \delta > 0$. Assume that for all $i\ge 0$, $x_{i+1} = A x_i + \eta_i$, where $x_0 = 0$, and $(\eta_i)_{i\ge 0}$ are i.i.d. zero-mean, $\sigma^2$-subgaussian random variables. We recall that $\Sigma = \frac{1}{n}\sum_{i=0}^{n-1} \Gamma_i(A)$. Then, with probability at least $1- \delta$, it holds 

\begin{align*}
    \left\Vert \frac{1}{n}\left(X \Sigma^{-\frac{1}{2}}\right)^\top  \left( X \Sigma^{-\frac{1}{2}} \right)  - I_{d_x} \right\Vert_2 \le \max( \varepsilon, \varepsilon^2), 
\end{align*}

provided $n \ge \frac{ 576(3+2\sqrt{2}) \sigma^4 \Vert \Gamma_\infty(A) \Vert_2^3 }{\varepsilon^2  \lambda_{\min}(\Sigma) } \left( \log\left(\frac{1}{\delta}\right) + \log(18) d_x\right)$.

\end{lemma}
We refer the reader to \cite{jedra2022finite} for a proof (see their Lemma 4) or \cite{ziemann2023tutorial} for a revised proof with exact constants (see their Lemma B.1). We further used Lemma \ref{lem:properties-Gamma} to upper bound $\Vert \sum_{i=0}^{n-1}\Gamma_i(A) \Vert_2$.  

Next, we present a relative perturbation bound on the spectrum of the empirical covariance matrix.
\begin{lemma}[Theorem 5 in \cite{barzilai2024simple}] \label{lem:perturbation}
Let $X \in \RR^{n \times d_x}$, with $n \ge d_x$. Let $\Sigma \in \RR^{d_x\times d_x}$ be a symmetric positive definite matrix. Define $\hat{\Sigma} = \frac{1}{n} X^\top  X$. For all $i \in [d_x]$, we have 

    \begin{align*}
        \vert \lambda_i(\hat{\Sigma}) - \lambda_i(\Sigma) \vert \le \lambda_{i}(\Sigma)  \left\Vert \frac{1}{n}\left(X \Sigma^{-\frac{1}{2}}\right)^\top  \left( X \Sigma^{-\frac{1}{2}} \right)  - I_{d_x} \right\Vert_2 
    \end{align*}
\end{lemma}

\begin{proof}[Proof of Theorem \ref{thm:cov-concentration-2}]
    The concentration result follows immediately by applying first Lemma \ref{lem:cov-concentration-2} with $\varepsilon = 1/2$, then using the relative perturbation bounds given by Lemma \ref{lem:perturbation}. The second statement of the Theorem is an immediate consequence of the first and simply follows by further using the inequality \emph{(iii)} of Lemma \ref{lem:properties-Gamma}.
\end{proof}

\subsection{On the $(\varepsilon,\delta,r)$-stability of \RLSE}

\begin{corollary}
    Let $\varepsilon>0$, $\delta\in(0,1)$ and $A$ such that $r\geq \mathrm{rank}(A)$. 
    Then:

    (i) Multivariate regression. Let $x_i\sim \mathcal{N}(0,\Sigma)$ with $\Sigma\succ 0$. For any number of samples $n\gtrsim d_x+\log(\frac{1}{\delta})$,  \RLSE~satisfies
    
    $$\mathbb{P}_A\left[\Vert \hat{A}-A\Vert_\F\lesssim r\sigma^2\frac{ d_x+d_y+\log(\frac{1}{\delta})}{n\underline{\lambda}^H_r(\Sigma)}\right]\geq 1-\delta, $$   
    and is $(\varepsilon,\delta,r)$-stable with sample complexity
    
    $$N\lesssim  r\sigma^2\frac{d_x+d_y+\log(\frac{1}{\delta})}{\varepsilon^2\underline{\lambda}^H_r(\Sigma)}.$$
    (ii) Linear system identification. For any number of samples $n$ verifying inequality (\ref{eq:lti_condition}), \RLSE~satisfies
    
    $$\mathbb{P}_A\left[\Vert \hat{A}-A\Vert_\F\lesssim r\sigma^2\frac{ d_x+\log(\frac{1}{\delta})}{n\underline{\lambda}^H_r(\Gamma_\infty(A))}\right]\geq 1-\delta, $$   
    and is $(\varepsilon,\delta,r)$-stable with sample complexity 
    
    $$N\lesssim  r\sigma^2\frac{d_x+\log(\frac{1}{\delta})}{\varepsilon^2\underline{\lambda}^H_r(\Gamma_\infty(A))}.$$ 
\end{corollary}

\begin{proof}
Let $\varepsilon>0$ and $\delta\in(0,1)$ and $A$ such that $r\geq \mathrm{rank}(A)$. We focus on the multivariate regression case (for system identification, just replace $\Sigma$ by $\Gamma_\infty(A)$). Let $n\geq 288(d_x+\log(\frac{1}{\delta}))$.

    Let $C_\RLSE>1$ the universal constant defined in the Proof of Theorem \ref{thm:rlse-gauss} such that given $n$ samples, one has with probability at least $1-\delta$,
    $$\Vert \hat{A}_n - A\Vert_\F^2\leq C_\RLSE \frac{r\sigma^2(d_x+d_y+\log(\frac{1}{\delta}))}{n\underline{\lambda}^H_r(\Sigma)}.$$


    Consider now 
    $$n=2C_\RLSE \frac{r\sigma^2(d_x+d_y+\log(\frac{1}{\delta}))}{\varepsilon^2\underline{\lambda}^H_r(\Sigma)},$$ 
    
    such that by Theorem \ref{thm:rlse-gauss} , we have $\Vert \hat{A}_n-A\Vert_\F^2\leq \varepsilon^2/2\leq \varepsilon^2$.
    
    Let now $A'\in\mathcal{C}(A,r,\varepsilon)$. Since $A'$ has rank at most $2r$ then given $n$ samples we also have by Theorem \ref{thm:rlse-gauss}, with probability at least $1-2\delta$,     
    $$\Vert \hat{A}_n-A'\Vert_\F^2\leq C_\RLSE\frac{2r\sigma^2(d_x+d_y+\log(\frac{1}{\delta}))}{n \underline{\lambda}^H_r(\Sigma)}\leq \varepsilon^2.$$

    Hence, we have just shown that \RLSE~ is $(\varepsilon,2\delta,r)$-stable  with $n$ samples (we can easily replace $2\delta$ by $\delta$ simply by adding universal constants in the definition of $n$). Since the sample complexity $N$ is defined as the minimum integer such that stability holds then we get that
    
    $$N\leq n\lesssim r\sigma^2\frac{d_x+d_y+\log(\frac{1}{\delta})}{\varepsilon^2\underline{\lambda}^H_r( \Sigma)}.$$

\end{proof}

\subsection{On the $(\varepsilon,\delta)$-stability of \TLSE}

\begin{corollary}
   Let $\varepsilon>0$, $\delta\in(0,1)$ and $A$ such that $r\geq \mathrm{rank}(A)$. Then:

    (i) Multivariate regression. Let $x_i\sim \mathcal{N}(0,\Sigma)$ with $\Sigma\succ 0$. For any number of samples $n\gtrsim d_x+\log(\frac{1}{\delta})$,  \TLSE~satisfies
    
     $$\mathbb{P}_A\left[\Vert \hat{A}_n-A\Vert_\F\lesssim \min_{k\in[r]}\left( k\sigma^2\frac{d_x+d_y+\log(\frac{1}{\delta})}{n\lambda_\mathrm{min}(\Sigma)}+ \sum_{i>k}s_i^2(A)\right)\right]\geq 1-\delta, $$   
    and is $(\varepsilon,\delta)$-stable with sample complexity    
    
    $$N\leq \min\left\{m: k^\star_{A,m}\sigma^2\frac{d_x+d_y+\log(\frac{1}{\delta})}{m\lambda_\mathrm{min}(\Sigma)}+ \sum_{i>k^\star_{A,m}}s_i^2(A)\leq c_5\varepsilon^2\right\},$$
    where $0<c_5<1$ is a universal constant.
    
    (ii) Linear system identification. For any number of samples $n$ verifying inequality (\ref{eq:lti_condition}), \TLSE~satisfies
    
    $$\mathbb{P}_A\left[\Vert \hat{A}_n-A\Vert_\F\lesssim \min_{k\in[r]}\left( k\sigma^2\frac{d_x+d_y+\log(\frac{1}{\delta})}{n\lambda_\mathrm{min}(\Gamma_\infty(A))}+ \sum_{i>k}s_i^2(A)\right)\right]\geq 1-\delta, $$   
    and is $(\varepsilon,\delta)$-stable with sample complexity    
    
    $$N\leq \min\left\{m: k^\star_{A,m}\sigma^2\frac{d_x+d_y+\log(\frac{1}{\delta})}{m\lambda_\mathrm{min}(\Gamma_\infty(A))}+ \sum_{i>k^\star_{A,m}}s_i^2(A)\leq c_6\varepsilon^2\right\},$$
    where $0<c_6<1$ is a universal constant.
\end{corollary}

\begin{proof}
    Let $\varepsilon>0$ and $\delta\in(0,1)$ and $A$ such that $r\geq \mathrm{rank}(A)$. We focus on the multivariate regression case (for system identification, just replace $\Sigma$ by $\Gamma_\infty(A)$). Let $n\geq 288(d_x+\log(\frac{1}{\delta}))$.

    Let $C_\TLSE>1$ the universal constant defined in the Proof of Theorem \ref{thm:tlse-gauss} such that given $n$ samples, one has with probability at least $1-\delta$,
    
    $$\Vert \hat{A}_n-A\Vert_\F^2\leq C_\TLSE\min_{k\in[r]}\left( k\sigma^2\frac{d_x+d_y+\log(\frac{1}{\delta})}{n\lambda_\mathrm{min}(\Sigma)}+ \sum_{i>k}s_i^2(A)\right).$$


    Consider  
    
    $$n= \min\left\{m: k^\star_{A,m}\sigma^2\frac{d_x+d_y+\log(\frac{1}{\delta})}{m\lambda_\mathrm{min}(\Sigma)}+ \sum_{i>k^\star_{A,m}}s_i^2(A)\leq \frac{\varepsilon^2}{2C_\TLSE}\right\}.$$

    By Theorem \ref{thm:tlse-gauss}, we have
    \begin{align*}
    \Vert \hat{A}_n-A\Vert_\F^2&\leq C_\TLSE\min_{k\in[r]}\left( k\sigma^2\frac{d_x+d_y+\log(\frac{1}{\delta})}{n\lambda_\mathrm{min}(\Sigma)}+ \sum_{i>k}s_i^2(A)\right)\\
    &\leq C_\TLSE\left( k^\star_{A,n}\sigma^2\frac{d_x+d_y+\log(\frac{1}{\delta})}{n\lambda_\mathrm{min}(\Sigma)}+ \sum_{i>k^\star_{A,n}}s_i^2(A)\right) \\
    &\leq \frac{\varepsilon^2}{2}\leq \varepsilon^2.    
    \end{align*}

    The previous computation shows also that $n$ must verify

    $$\Vert A-\Pi_{k^\star_{A,n}}(A)\Vert_\F\leq\varepsilon$$

    Let now $A'\in \mathcal{D}(A,n,\varepsilon)$ and denote for simplicity $K=k^\star_{A,n}$. 
    
    $A'$ can also be rewritten as following
    
    $$A'=\Pi_{K}(A)+\frac{2C\varepsilon}{\sqrt{K}}QW^\top_{-K} + A-\Pi_K(A),$$ 
    
    for some $Q\in\mathcal{P}_{K}^{d_y}$  (replace $QW^\top_{-K}$ by $U_KR^\top$ for some $R\in\mathcal{P}_{K}^{d_x}$ in the other possible case).  
    
    We have the following  result on singular values, see Theorem 2 of \cite{fan1951maximum}: For any matrices $P,Q$ and integers $(i,j)$, 
    
    $$
    s_{i+j-1}(P+Q)\leq s_i(P)+s_j(Q).
    $$
    
    Therefore, for any $i=1,...,r-K$ ($A'$ has rank at most $r+K$),   
    
     \begin{align*}
         s_{2K+i}(A')&\leq s_{2K+1} (\Pi_{K}(A)+\frac{2C\varepsilon}{\sqrt{K}}QW^\top_{-K} )+s_i(A-\Pi_K(A))\\
         &=s_i(A-\Pi_K(A))\\
         &=s_{K+i}(A).
     \end{align*}

     Given $n$ samples, we also have by Theorem \ref{thm:tlse-gauss}, with probability at least $1-2\delta$,
    
     \begin{align}\label{eq:stability_tlse}
     \begin{split}
         \Vert \hat{A}_n -A'\Vert_\F^2&\leq C_\TLSE\min_{k\in[r]}\bigg[ 2k\sigma^2\frac{d_x+d_y+\log(\frac{1}{\delta})}{n\lambda_\mathrm{min}(\Sigma)}+ \sum_{i>k}s_i^2(A')\bigg]
         \\&\leq C_\TLSE\left(2K\sigma^2\frac{d_x+d_y+\log(\frac{1}{\delta})}{n\lambda_\mathrm{min}(\Sigma)}+ \sum_{i>2K}s_i^2(A')\right) \\
         &\leq 2C_\TLSE\left(K\sigma^2\frac{d_x+d_y+\log(\frac{1}{\delta})}{n\lambda_\mathrm{min}(\Sigma)}+\sum_{i>K}s_i^2(A)\right)\\ 
         &\leq \varepsilon^2.
    \end{split}
    \end{align}

 Hence, we have just shown that \TLSE~ is $(\varepsilon,2\delta)$-stable with $n$ samples (again, we can easily replace $2\delta$ by $\delta$ by adding universal constants in the definition of $n$). Since the sample complexity $N$ is defined as the minimum integer such that stability holds then we get that

 $$N\leq n =\min\left\{m: k^\star_{A,m}\sigma^2\frac{d_x+d_y+\log(\frac{1}{\delta})}{m\lambda_\mathrm{min}(\Sigma)}+ \sum_{i>k^\star_{A,m}}s_i^2(A)\leq \frac{\varepsilon^2}{2C_\TLSE}\right\}.$$
\end{proof}

\begin{remark}
    Finally, note that these results also hold for other rank-constrained and rank-adaptive algorithms, for instance those of \cite{bunea2011optimal} or the nuclear norm estimator in Appendix \ref{app:sec5}, since they enjoy similar but looser upper bounds (up to model-dependent constants).
\end{remark}

\newpage
\section{Numerical experiments}\label{app:sec7}

In this section, we present further numerical experiments to illustrate the theoretical results obtained in Sections \ref{sec:denoising} and \ref{sec:lse}. We first compare our denoising lemma (derived in Section \ref{sec:denoising}) to Lemma 3.5 of \cite{chatterjee2015matrix}, using a simple synthetic example. Next, we study the performance of our estimators \RLSE~ and \TLSE, and investigate the tightness of the error upper bound of \TLSE\ derived in Theorem  \ref{thm:tlse-gauss}. Finally, we describe the computational resources used for our experiments. 

\subsection{Matrix denoising lemmas}

We first restate Lemma 3.5 in \cite{chatterjee2015matrix}: for $Z=\bar{A}-A$,

$$
\|\bar{A}\left((1+\tau)\|Z\|_2\right)-A\Vert_\F^2\leq f(\tau) \Vert Z\Vert_2 \Vert A\Vert_1
$$

where $f(\tau) = ((4+2\tau)\sqrt{\frac{2}{\tau}}+\sqrt{2+\tau})^2$.

In comparison, our Lemma \ref{lem:trade_off_k} yields

$$\Vert \bar{A}\left((1+\tau)\|Z\|_2\right)-A\Vert_\F^2\leq 18(k(\tau)\Vert Z\Vert_2^2+ \sum_{i>k(\tau)}s_i^2(A))$$
where $k(\tau)=\max\{i:s_i(\bar A)\geq (1+\tau)\Vert Z\Vert_2\}$. We are interested in comparing both upper bounds with respect to the parameter $\tau$.  $A$ is a $50\times 50$ matrix of rank $r=10$ with entries sampled uniformly at random in $[-1,1]$. We then generate $\bar A = A+Z$ where the entries of $Z$ are i.i.d sampled from ${\cal N}(0,1)$. Given  $Z$ and $A$, we evaluate the two upper bounds for several values of $\tau$, and average them over $T=100$ runs.

Figure \ref{fig:denoising_bounds} depicts the behavior of both upper bounds as a function of $\tau$. We can observe that our upper bound is indeed smaller, and stays bounded between $r\Vert Z\Vert_2^2$ (corresponding to $k(\tau)=r$) and $\Vert A\Vert_\F^2$ (corresponding to $k(\tau)=0$).
\begin{figure}[H]
    \centering
    \includegraphics[width=0.5\columnwidth]{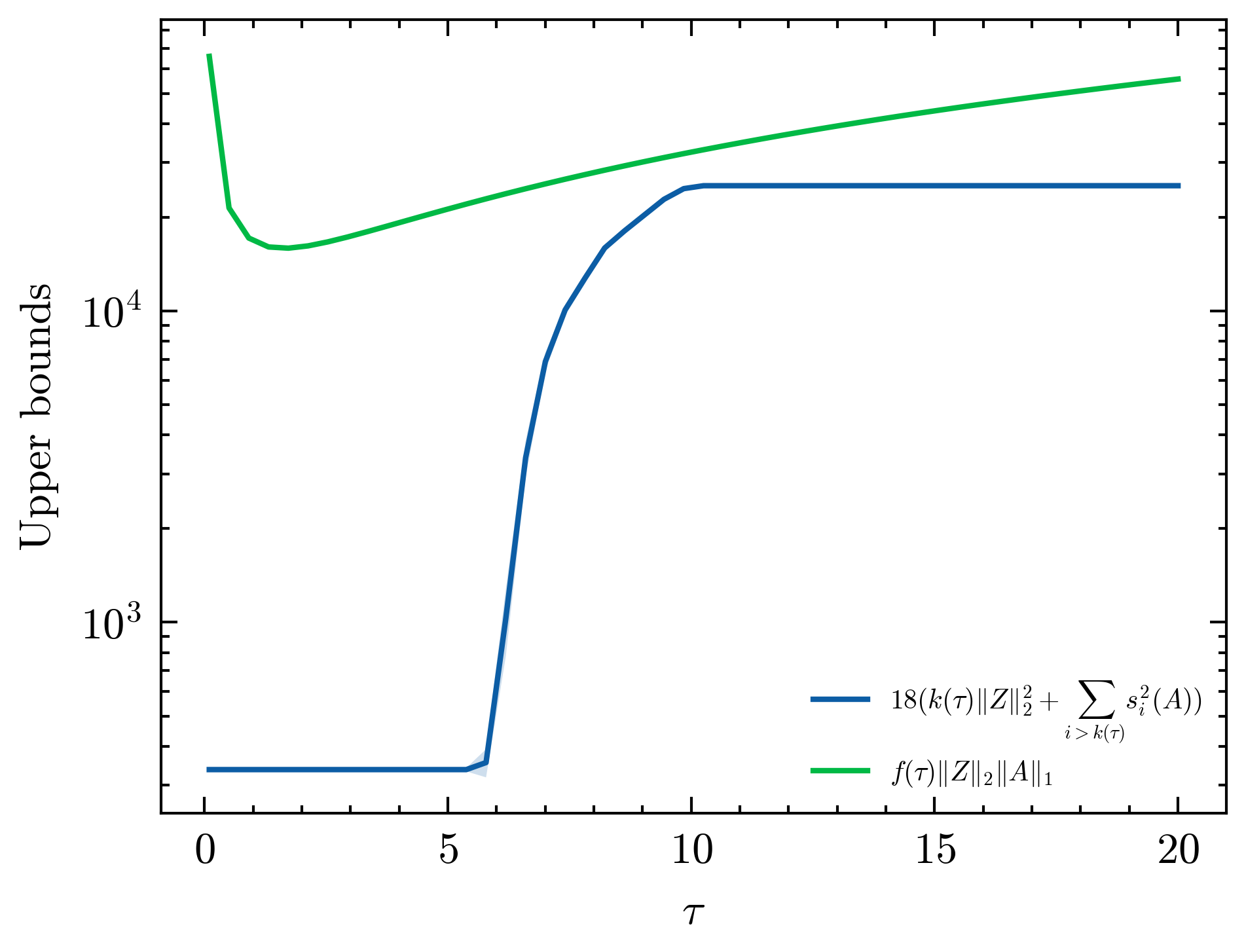}    
    \vspace*{-3mm}
    \caption{Comparison of Lemma \ref{lem:trade_off_k} with the upper bound from \cite{chatterjee2015matrix} as a function of the threshold $\tau$.}
    \label{fig:denoising_bounds}
\end{figure}

 \subsection{On the importance of adaptivity}\label{subsec:exp-p}

We consider the same experiments as those already presented in Section \ref{sec:exp}, which we re-describe briefly. We estimate $A\in\mathbb{R}^{d\times d}$ ($d=50$), of rank $r$, from a linear model $Y=XA+E$. 1) Entries of $A$ are first sampled uniformly at random in $[0,1]$. 2) We compute its SVD $A=USV^\top$. 3) We keep its singular vectors but change its singular values as follows: $s_j(A)=\frac{1}{j^b}$ where $b$ is a parameter acting as a proxy for signal-to-noise ratio. In Section \ref{sec:exp}, we considered the case where the rank of $A$ was $r=10$. Here we investigate the high-rank regime with $r=45,\sigma=0.4$. Figure~\ref{fig:fro_rank_high} shows the performance of \RLSE{} and its adaptive counterparts, \TLSE{} and RSC, as a function of $b$, using $n = 1000$ samples. Similar to the low-rank regime, we observe that as the noise level increases, both \TLSE{} and RSC adapt their effective rank to concentrate on the singular subspaces that can be reliably estimated from the noisy data. 

\begin{figure}
    \includegraphics[width=0.485\columnwidth]{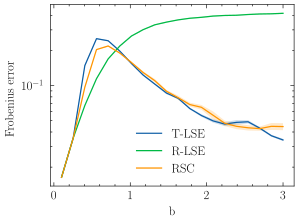}
    \hspace{0.3cm}
    \includegraphics[width=0.47\columnwidth]{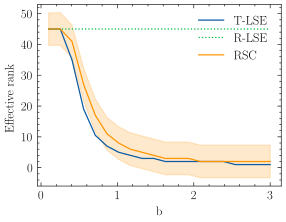}
    \vspace*{-3mm}
    \caption{Multivariate regression: Frobenius error (left) and effective rank (right) vs noise level $b$, $(r=45)$.}
    \label{fig:fro_rank_high}
\end{figure}

\subsection{System identification}

We further give an example of rank-adaptive estimation for the case of linear system identification. We first construct a stable matrix $A\in\mathbb{R}^{d\times d}$ with $d=50$ (stable means such that $\rho(A)<1$). To do so, we consider a symmetric positive definite $A$ constructed as follows. 
1) We first sample its entries uniformly at random in $[0,1]$. 2) We compute its SVD $A=USV^\top$. 3) We change $A$ to $USU^\top$. 4) We change the entries of $S$ to $(\frac{1}{(j+1)^b})_{j=1}^r$. For this choice of $A$, the system is stable. We assume that $x_{t+1}=Ax_t+\mathcal{N}(0,\sigma^2 I_d)$ and $x_0=0$ with $\sigma=0.1$.  Similarly to the Gaussian input case, Figure \ref{fig:lti_exp} shows  that adaptiveness becomes preferable than its rank-constrained counterpart as $b$ increases and the signal to noise ratio decreases. 

\begin{figure*}[b]
    \includegraphics[width=0.485\textwidth]{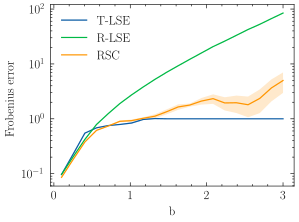}
    \hspace{0.3cm}
    \includegraphics[width=0.47\textwidth]{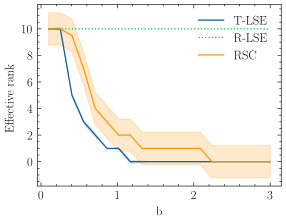}
    \vspace*{-3mm}
    \caption{System identification: Frobenius error (left) and effective rank (right) vs noise level $b$, $(r=10)$.}
    \label{fig:lti_exp}
\end{figure*}

\subsection{Tightness of our upper bound}\label{sec:additional_exp}

In this section, we investigate the tightness of the error upper bound for \TLSE{} established in Theorem~\ref{thm:tlse-gauss}. Specifically, we aim to determine whether the bound accurately captures the dependence of the true error $\|\hat A_n - A\|_\F^2$ (and not the relative error anymore) on the underlying problem parameters.

For sake of simplicity, we consider the multivariate regression setting. $A$ is a $(d,d)$-square matrix, with rank $r$. Its entries are sampled uniformly at random in $[0,1]$. The noise matrix $E$ entries are i.i.d sampled from $\mathcal{N}(0,1)$. The covariates $x_i$ are i.i.d sampled from a multivariate Gaussian distribution $\mathcal{N}(0,I_d)$. The outputs $y_i$ are the rows of $Y=XA+E$. We study the performances of \TLSE~ as a function of 1) number of samples $n$, 2) dimension $d$.  We compare them to the upper bound given by Theorem \ref{thm:tlse-gauss}.   

For each experiment, i.e., for each choice of $(n,d)$, we sample $A,X$ once and then sample $T$ times $E$ to obtain $Y$ and $\hat A_n$. We compute the identification error $\Vert \hat A_n - A\Vert_\F^2$ in each trial and output the average. 

\begin{enumerate}
 
\item {As a function of $n$ ($d=50$, $r=10$).} The true error of \TLSE\ and our upper bound are close. Both exhibit a hyperbolic decrease towards zero as can be seen in Figure \ref{fig:tlse_n}.


\begin{figure}[H]
    \centering
        \includegraphics[width=0.5\linewidth]{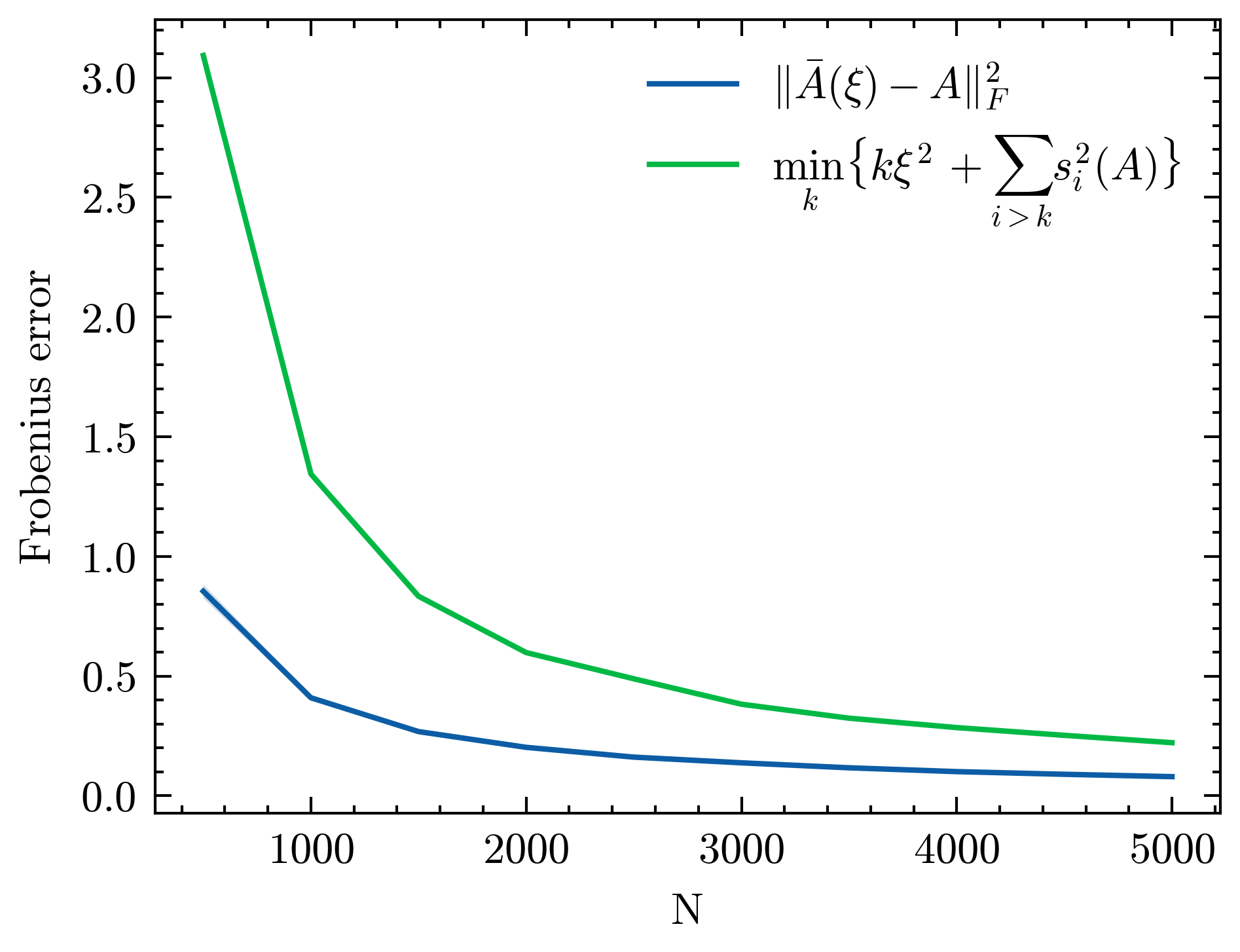}
    \vspace{-3mm}
    \caption{Performance of \TLSE~and its corresponding error upper bound as a function of $n$.}
    \label{fig:tlse_n}
\end{figure}

\item {As a function of $d$ ($n=5000$, $r=10$).} We observe in Figure \ref{fig:tlse_d} that the \TLSE\ error is linear in $d$ as predicted by the upper bound (and lower bound). For the upper bound, the observed deviation from perfect linearity arises from the degradation of the accuracy of the empirical eigenvalue $\underline{\lambda}_r^H(\hat \Sigma)$ as $n$ is kept fixed.


\begin{figure}[H]
    \centering
        \includegraphics[width=0.5\linewidth]{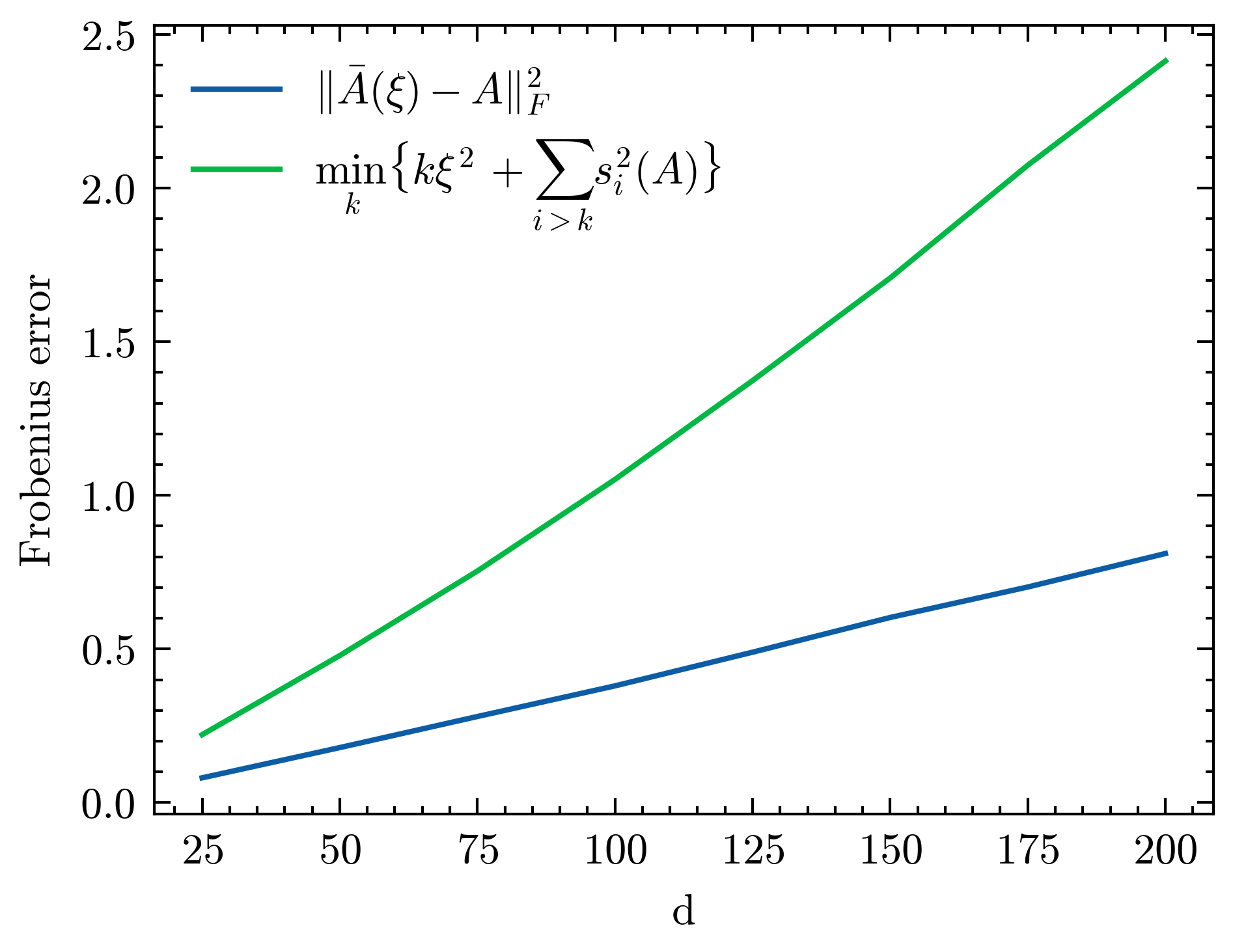}
    \vspace{-3mm}
    \caption{Performance of  \TLSE~ and its corresponding error upper bound as a function of $d$.}
    \label{fig:tlse_d}
\end{figure}


\end{enumerate}

Overall, our error upper bound reflects the dependence of \TLSE{}'s performance on the system parameters, but it is not always tight. Deriving sharper bounds remains an open direction for future work. In particular, an interesting step would be to replace, in our Theorems \ref{thm:tlse-gauss} and \ref{thm:tlse-lti}, $\lambda_{\min}(\hat{\Sigma})$ by $\underline{\lambda}_k^H(\hat \Sigma)$.



\subsection{Computing resources}\label{app:resources}

All experiments were performed on a 12th Gen Intel(R) Core(TM) i7-1280P   1.80 GHz, with 64 Go of available RAM memory. 

Coding was done on Python and standard numerical analysis libraries (e.g NumPy and Matplotlib) were used. The publicly available SciencePlots, \url{https://github.com/garrettj403/SciencePlots} library, in its Version 1.0.9, was used for plots.

\end{document}